\documentclass[10pt, conference, letterpaper]{IEEEtran}

\usepackage{amsthm,amssymb,graphicx,multirow,amsmath,color,amsfonts,balance}
\usepackage{bm}
\usepackage[latin1]{inputenc}
\usepackage[update,prepend]{epstopdf}
\usepackage[noadjust]{cite}
\usepackage{mathtools}
\usepackage{multirow}
\usepackage{bbm} 
\usepackage{pdfpages}
\usepackage{tabulary}
\usepackage{multirow}
\usepackage{comment}
\usepackage{algorithm}
\usepackage{algorithmicx}
\usepackage{algpseudocode}
\usepackage{etoolbox}
\usepackage{enumerate}
\usepackage{caption}
\usepackage{subcaption}

\usepackage[algo2e,ruled,vlined,linesnumbered]{algorithm2e}

\SetCommentSty{mycommfont}

\usepackage[svgnames]{xcolor}
\def\beq{\begin{equation}}
\def\eeq{\end{equation}}
\def\beqa{\begin{eqnarray}}
\def\eeqa{\end{eqnarray}}
\def\beqan{\begin{eqnarray*}}
\def\eeqan{\end{eqnarray*}}
\include{notation_new}

\usepackage{cite}

\usepackage[normalem]{ulem} 

\newcommand{\citep}[1]{\cite{#1}}

\newcounter{newenumi}
\renewcommand{\tilde}{\widetilde}
\renewcommand{\hat}{\widehat}

\def\beq{\begin{equation}}
\def\eeq{\end{equation}}
\def\beqa{\begin{eqnarray}}
\def\eeqa{\end{eqnarray}}
\def\beqan{\begin{eqnarray*}}
\def\eeqan{\end{eqnarray*}}

\DeclareMathOperator{\diag}{Diag}

\newtheorem{theorem}{Theorem}
\newtheorem{lemma}{Lemma}

\theoremstyle{definition}

\def\Exp{\mathbb{E}}

\newcommand{\Nrf}{N_{\rm RF}}

\newcommand{\Pup}{P^{\rm u}}

\newcommand{\Pdl}{P^{\rm d}}
\newcommand{\SINRup}{{\rm SINR}^{\rm u}}
\newcommand{\SINRdl}{{\rm SINR}^{\rm d}}
\newcommand{\Ckkp}{\Cbf_{{kk'}}}
\newcommand{\Ckpk}{\Cbf_{{k'k}}}
\newcommand{\Ckk}{\Cbf_{{k}}}
\newcommand{\Ckpkp}{\Cbf_{{k'}}}
\newcommand{\hhk}{\hat{\hbf}_{{k}}}

\newcommand{\hk}{{\hbf}_{{k}}}
\newcommand{\hhkp}{\hat{\hbf}_{{k'}}}
\newcommand{\hhkpp}{\hat{\hbf}_{{k''}}}

\newcommand{\hdek}{{\hbf}^{{\rm d}}_{k}}

\newcommand{\zero}{\mathbf{0}}

\newcommand{\cbf}{\mathbf{c}}

\newcommand{\ebf}{\mathbf{e}}

\newcommand{\hbf}{\mathbf{h}}
\newcommand{\pbf}{\mathbf{p}}

\newcommand{\ubf}{\mathbf{u}}

\newcommand{\vbf}{\mathbf{v}}

\newcommand{\Tbf}{\mathbf{T}}

\newcommand{\xbf}{\mathbf{x}}

\newcommand{\ybf}{\mathbf{y}}

\newcommand{\zbf}{\mathbf{z}}

\newcommand{\Abf}{\mathbf{A}}
\newcommand{\Bbf}{\mathbf{B}}

\newcommand{\Cbf}{\mathbf{C}}
\newcommand{\Dbf}{\mathbf{D}}

\newcommand{\Gbf}{\mathbf{G}}
\newcommand{\Hbf}{\mathbf{H}}

\newcommand{\Ibf}{\mathbf{I}}
\newcommand{\Jbf}{\mathbf{J}}

\newcommand{\Rbf}{\mathbf{R}}

\newcommand{\Sbf}{\mathbf{S}}
\newcommand{\Ubf}{\mathbf{U}}

\newcommand{\Wbf}{\mathbf{W}}

\newcommand{\Ybf}{\mathbf{Y}}
\newcommand{\Zbf}{\mathbf{Z}}

\newcommand{\tran}{^{\text{\sf T}}}
\newcommand{\trans}{{\text{\sf T}}}
\newcommand{\herm}{^{\Htran}}
\newcommand{\herms}{\Htran}

\newcommand{\vect}[1]{\mathbf{#1}}

\def\diag{\mathrm{diag}}

\def\kron{\otimes}
\def\tr{\mathrm{tr}}

\def\Htran{\mbox{\tiny $\mathrm{H}$}}

\def\CN{\mathcal{CN}} 

\newtheorem{Theorem}{Theorem}
\newtheorem{Proposition}{Proposition}

\newtheorem{Definition}{Definition}

\begin{document}
\bstctlcite{IEEEexample:BSTcontrol}
\title{Cell-Free Massive MIMO with\\ Low-Complexity Hybrid Beamforming}

\author{Abbas Khalili$^1$, Alexei Ashikhmin$^2$, Hong Yang$^2$\\
$^1$NYU Tandon School of Engineering, $^2$Nokia Bell Labs\\
Emails: ako274@nyu.edu, \{alexei.ashikhmin, h.yang \}@nokia-bell-labs.com}

\maketitle
\begin{abstract}
Cell-Free Massive Multiple-input Multiple-output (mMIMO) consists of many access points (APs) in a coverage area that jointly serve the users. These systems can significantly reduce the interference among the users compared to conventional MIMO networks and so enable higher data rates and a larger coverage area. However, Cell-Free mMIMO systems face multiple practical challenges such as the high complexity and power consumption of the APs' analog front-ends. Motivated by prior works, we address these issues by considering a low complexity hybrid beamforming framework at the APs in which each AP has a limited number of RF-chains to reduce power consumption, and the analog combiner is designed only using the large-scale statistics of the channel to reduce the system's complexity. We provide closed-form expressions for the signal to interference and noise ratio (SINR) of both uplink and downlink data transmission with accurate random matrix approximations. Also, based on the existing literature, we provide a power optimization algorithm that maximizes the minimum SINR of the users for uplink scenario. Through several simulations, we investigate the accuracy of the derived random matrix approximations, trade-off between the $95\%$ outage data rate and the number of RF-chains, and the impact of power optimization. We observe that the derived approximations accurately follow the exact simulations and that in uplink scenario while using MMSE combiner, power optimization does not improve the performance much. 
\end{abstract}

\section{Introduction}

Access Points (APs) with coherent transmission can increase the received power without requiring additional transmit power and reduce the interference among the user's signals leading to increased Signal to Interference and Noise Ratio (SINR) \cite{shamai2001enhancing,bjornson2020scalable,ngo2017cell,nayebi2017precoding,bjornson2019making}. In the context of Massive Multiple-input and Multiple-output (mMIMO) systems, this coherent transmission is referred to as Cell-Free mMIMO, where a network of interconnected APs simultaneously serve the users over a designated area \cite{ngo2017cell,nayebi2017precoding} an illustration of such networks is shown in Fig.~\ref{fig:cfm}. 

Cell-Free mMIMO systems have attracted a lot of attention in the literature for potential deployment in the next generations of wireless networks which are envisioned to operate at high frequencies such as Millimeter Wave \cite{femenias2019cell, rangan2014millimeter}. Data transmission at high frequencies suffer from high path loss due to propagation characteristics of the channel at high frequencies. To mitigate the high path-loss, the transceivers need to use large antenna arrays which has lead to high power consumption at the transceivers. This is a major obstacle for practical implementation of these systems \cite{rangan2014millimeter}. To reduce the power consumption, it is suggested to use hybrid transceivers in which, the antennas are connected to a few RF-chains through a network of phase shifters and/or analog switches \cite{el2014spatially,gao2016energy}.   

\begin{figure}[t]
\centering
    \includegraphics[width=0.7\linewidth]{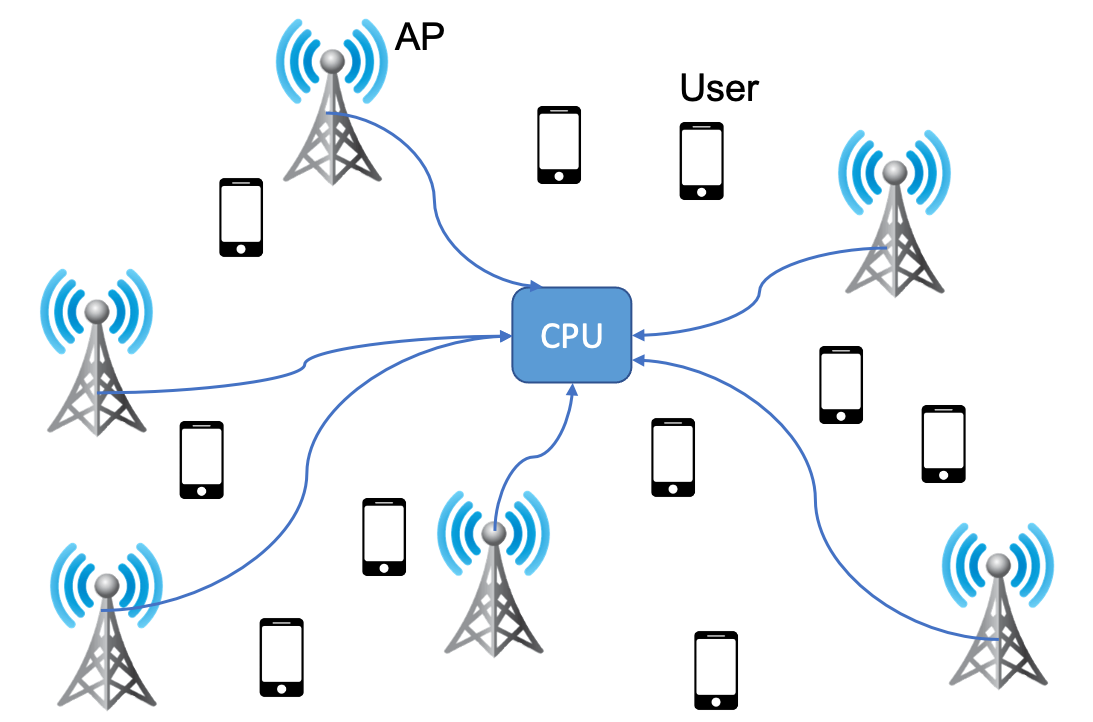}
    \caption{Cell free Massive MIMO system.}
    \label{fig:cfm}
\end{figure}

There is a large literature on the Cell-Free mMIMO. For example, \cite{yan2021scalable} considers Cell-Free mMIMO system in the context of internet of things where each AP only has one antennas. They provide the uplink and downlink per user equipment SINR expressions considering Minimum Mean Square Error (MMSE) combining and maximum ratio percoding, respectively. Furthermore, they derive accurate random matrix based approximations of the uplink SINR expression which only depend on the large scale statistics of the system along with power optimization methods for uplink and downlink scenarios to maximize the minimum SINR among the users. 
The paper \cite{femenias2019cell} considers a Cell-Free mMIMO systems in the context of mmWave systems where the APs use hybrid APs and provides hybrid beamforming method based on large scale fading coefficients of the channel along with SINR expressions for uplink and donwlink scenarios considering zero-forcing precoding and combiner, respectively. 

Here, we consider a Cell-Free mMIMO system with hybrid APs. At each AP we devise a hybrid beamfoming scheme in which the analog beamformer is designed based on the large scale statistics of the channel, and then after estimating the resulting effective channels between the APs and users, the digital beamfomer is designed based on the estimated channel using MMSE combiner for uplink and RZF precoder for downlink. We provide per user SINR expressions for both uplink and downlink scenarios along with accurate random matrix approximations which only depend on large scale statistics of the system (Sec.~\ref{sec:up} and Sec.~\ref{sec:dl}). Furthermore, we provide optimality achieving power optimization method for the uplink scenario that maximizes the minimum SINR among the users (Sec.~\ref{sec:up}). Finally, we provide various simulations of practical systems to validate the accuracy of the derived approximations and investigate the impact of the number of RF-chains, the considered hybrid beamforming method, and power optimization on the system's perfomance (Sec.~\ref{sec:sim}).

\section{System Model}
\subsection{Cell-Free Network Model}
\label{subsec:network-model}

We consider a Cell-Free mMIMO network with $M$ APs, each equipped with $N$ antennas, where all APs are connected via fronthaul connections to a Central Processing Unit (CPU). There are $K$ single-antenna users in the network and the channel between AP $m$ and user $k$ is denoted by $\hbf_{mk} \in \mathbb{C}^N$. We use block fading channel model and assume that $\hbf_{mk}$ is constant over time-frequency blocks of $\tau_c$ channel uses and in each block, $\hbf_{mk}$ can be modeled as an independent realization of a correlated Rayleigh fading distribution:
\begin{equation}
\label{eq:chg}
\hbf_{mk} \sim \CN(\vect{0}, \Rbf_{mk})
\end{equation}
where $\Rbf_{mk} \in \mathbb{C}^{N \times N}$ is the spatial correlation matrix which we assume is known at the APs. 

As shown in Fig.~\ref{fig:hbf}, we assume that each AP performs hybrid beamforming using $\Nrf$ RF-chains ($\Nrf\leq N$). To reduce the system complexity, and communication overhead, we design the analog precoder network at each AP (matrices $\Wbf_m, ~m \in[M]$) using only the spatial correlation matrices available at the APs \footnote{We use the notation $[N]$ to denote the set $\{1,2,\ldots,N\}$}. After designing the analog precoder, we will have an effective channel of dimension $\Nrf$ between each user and AP. We consider a time division duplexing system, where each time-frequency fading block consists of $\tau_p $,  $\dfrac{\tau_c-\tau_p}{2}$, and $\dfrac{\tau_c-\tau_p}{2}$ channel uses dedicated for channel estimation of the effective channel, uplink transmission, and downlink transmission, respectively. Analog precoder design and channel estimation are discussed below and uplink and downlink data transmission along with our results are discussed in Sec.~\ref{sec:up} and Sec.~\ref{sec:dl}, respectively.

\subsection{Analog Precoder Design Using Large Spatial Correlation}
\label{subsec:ABF}

 To the best of our knowledge, there are only two schemes on design of the analog precoder based on the spatial correlation matrices \cite{park2017exploiting,zhu2016novel}. The analog beamforming method proposed in \cite{zhu2016novel} aims to maximize the sumrate of the users, but here we are interested in max-min fairness and want to maximize the minimum rate of the users. Therefore, similar to \cite{mai2018two}, we use eigen beamforming method \cite{park2017exploiting}. 
 
In the eigen beamforming method, each RF-chain is allocated to a user and its corresponding analog combiner vector is derived as follows. 
Assume RF-chain $i$ at AP $m$ is allocated to the user $k$ and let us denote its analog precoder column with $\vect{w}_{m,i}$. Using the eigen value decomposition of $\Rbf_{mk} =\Ubf^{\herms}_{mk}\Lambda_{mk} \Ubf_{mk}$, we have
\begin{align}
    \label{eq:wg}
    \vect{w}_{m,i} = \frac{1}{N}e^{j\angle {\overline{\ubf}}_m},
\end{align}
where ${\overline{\ubf}}_{m}$ is the eigen vector corresponding to the largest eigne value of $\Rbf_{mk}$ and $\angle \cdot$ is an element-wise operator returning the angle of each element of its input. If multiple RF-chains are assigned to a user, we use the eigne vectors corresponding to that many largest eigen values, each for an RF-chain.

As mentioned above, using the eigen beamforming method requires allocating each RF-chain to a user. We have a total of $M \Nrf$ RF-chains that need to be allocated among the users. However,  finding the optimal user allocation is not a tractable problem. So, we use the heuristic algorithm proposed in \cite{femenias2019cell} to assign the RF-chains. This method aims to maximize the minimum sum of the average energy of the effective channels (channel after analog precoding) to the users. Since each RF-chain is allocated to a user, we assume $K \leq M \Nrf$ so each user is allocated at least one RF-chain.

\begin{figure}[t]
\centering
    \includegraphics[width=0.6\linewidth]{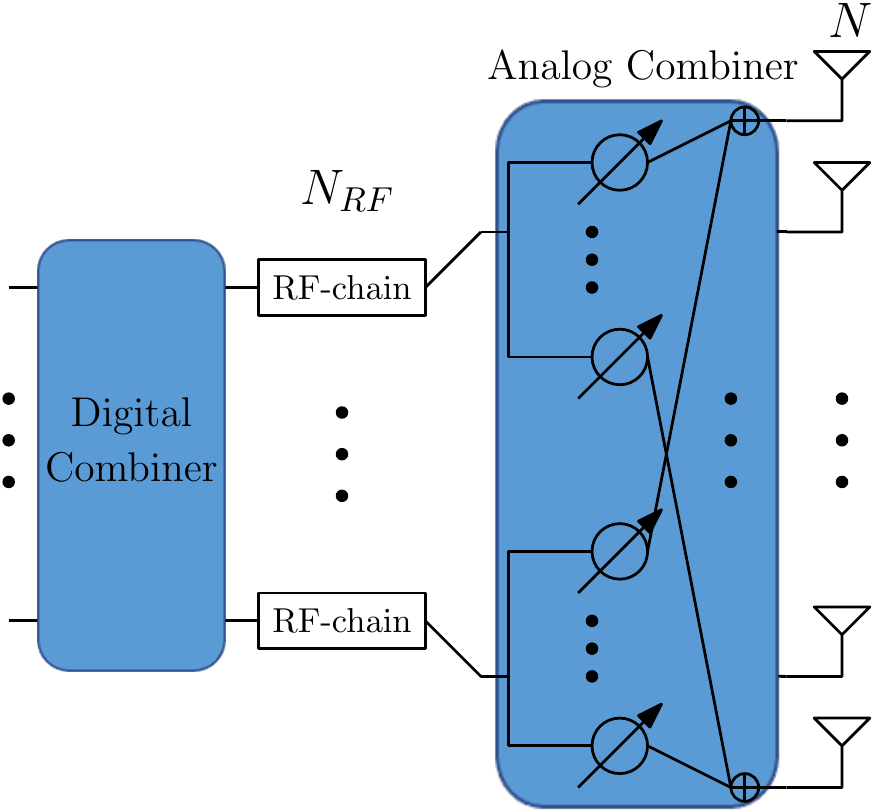}
    \caption{Hybrid Beamforming at each AP.}
    \label{fig:hbf}
\end{figure}

\subsection{Effective Channel Estimation}
\label{subsec:eche}

For channel estimation, we assume that user $k$ is assigned a unique random pilot ${\psi}_k \in \mathbb{C}^{\tau}$ with $\|\psi_k\|^2 = \tau_p$ which is transmitted through $\tau_p$ channel-uses. At AP $m$, the received signal is linearly combined into $\Nrf$ streams using the analog precoder $\Wbf_{m}$ and then forwarded to the CPU for channel estimation.
Let us denote the received vector at AP $m$ and channel-use $i$ by  $\ybf_{mi}\in \mathbb{C}^{\Nrf}$, we have
\begin{align}
\label{eq:rAPm}
    \Ybf_m = \sqrt{P_p} \Wbf_m^{\herms}\Hbf_m \vect{\Psi}\tran + \Wbf_m^{\herms}\Zbf_m,
\end{align}
where $P_p$ is the transmit power of the pilots, $\Ybf_m = [\ybf_{m1},\ybf_{m2}, \ldots, \ybf_{m\tau_p}]$, $\Hbf_m = [\hbf_{m1},\hbf_{m2}, \ldots, \hbf_{mK}]$, $\Psi = [\psi_1,\psi_2, \ldots, \psi_K]$, and $\Zbf_m = [\zbf_{m1},\zbf_{m2}, \ldots, \zbf_{m\tau_p}]$ is the noise matrix whose elements are i.i.d. $\CN(0,1)$. Following basic arithmetic, we can rewrite \eqref{eq:rAPm} as follows 
\begin{align}
    \ybf_{{\rm e}m} = \sqrt{P_p} \left(\vect{\Psi} \kron \Ibf_{\Nrf} \right)\hbf_{{\rm e}m} + \zbf_{{\rm e}m},
\end{align}
where $\ybf_{{\rm e}m} = [\ybf^{\trans}_{m1},\ybf^{\trans}_{m2},\ldots, \ybf^{\trans}_{m\tau_p}]\tran \in \mathbb{C}^{\tau_p\Nrf}$, $\hbf_{{\rm e}m} = $ $[(\Wbf^{\herms}_m\hbf_{m1})\tran,$ $(\Wbf^{\herms}_m\hbf_{m2})\tran,$ $\ldots,$  $(\Wbf^{\herms}_m\hbf_{mK})\tran]\tran \in \mathbb{C}^{K\Nrf}$, and $\zbf_{{\rm e}m} = [(\Wbf^{\herms}_m\zbf_{m1})\tran,(\Wbf^{\herms}_m\zbf_{m2})\tran,\ldots, (\Wbf^{\herms}_m\zbf_{m\tau_p})\tran]\tran$. Let us denote $  \vect{\Psi}_{\rm e} = \vect{\Psi} \kron \Ibf_{\Nrf} $. Using LMMSE channel estimation which minimizes the MSE distortion, we can estimate the vector $\hbf_{{\rm e}m}$ as
\begin{align}
\label{eq:chev}
    \hat{\hbf}_{{\rm e}m} = \sqrt{P_p} {\Rbf}_{{\rm e}m} \vect{\Psi}_{\rm e}\herm \left(\vect{\Psi}_{\rm e}{\Rbf}_{{\rm e}m} \vect{\Psi}_{\rm e}\herm + \Cbf_{{\rm e}zm} \right)^{-1} \ybf_{{\rm e}m},
\end{align}
where 
\begin{subequations}
\begin{align}
&\Rbf_{{\rm e}m} = \Exp \left[ \hbf_{{\rm e}m} \hbf^{\herms}_{{\rm e}m} \right] = \diag\{\Rbf_{{\rm e}mk}\}_{k\in [K]},\\
&{\Rbf}_{{\rm e}mk}  = \Wbf^{\herms}_m\Rbf_{mk} \Wbf_{m},\\
&\Cbf_{{\rm e}zm} = \Exp \left[ \zbf_{{\rm e}m} \zbf^{\herms}_{{\rm e}m} \right] =  \Ibf_{K} \kron \Wbf^{\herms}_m\Wbf_{m},
\end{align}
\end{subequations}
where we use the notation $\diag\{\Abf_i\}_{i\in[N]}$ to denote a block diagonal matrix whose $i^{\rm th}$ matrix on the diagonal is $\Abf_i$. Based on \eqref{eq:chev}, the estimated channel between user $k$ and AP $m$ is 
\begin{align}
    \label{eq:chmkh}
    \begin{aligned}
    &\hat{\hbf}_{{\rm e}mk}\!=\sqrt{P_p} {\Rbf}_{{\rm e}mk}\left( \vect{\psi}_k \herm \kron \Ibf_{N_{RF}} \right)\\
    & \left(\vect{\Psi}_{\rm e}\Rbf_{{\rm e}m} \vect{\Psi}_{\rm e}\herm + \Cbf_{{\rm e}zm} \right)^{-1} \!\!\ybf_{{\rm e}m}.  
    \end{aligned}
\end{align}
For our analysis, we are interested in the cross-covariance matrix of the estimated effective channels between CPU and users $k$ and $k'$ which is provided in the following proposition. 

\begin{Proposition}(\textbf{Covariance Matrix of Estimated Channel})
Consider the effective channel between user $k$ and CPU $\hbf_{{\rm e}k} = [\hbf_{{\rm e}1k}^{\trans},\hbf_{{\rm e}2k}^{\trans}, \ldots, \hbf_{{\rm e}Mk}^{\trans} ]\tran$, for users $k$ and $k'$, we have
\begin{align}
&\Cbf_{{\rm e}kk'} =\diag\{\Cbf_{{\rm e}mkk'} \}_{m \in [M]},\\
&\begin{aligned}
    \Cbf_{{\rm e}mkk'} &= P_p {\Rbf}_{{\rm e}mk} \left( \vect{\psi}_k\herm \kron \Ibf_{N_{RF}} \right)\times \\
    &\left(\vect{\Psi}_{\rm e}\Rbf_{{\rm e}m} \vect{\Psi}_{\rm e}\herm + \Cbf_{{\rm e}zm} \right)^{-1} \left( \vect{\psi}_{k'} \kron \Ibf_{N_{RF}} \right){\Rbf}_{{\rm e}mk'}
    \end{aligned}
\end{align}
\end{Proposition}
\begin{proof}
From \eqref{eq:chmkh}, for the cross-covariance matrix of the estimate effective channel between user $K$ and AP $m$ and user ${k'}$ and AP ${m'}$ is 
\begin{align}
\label{eq:expmkmkp}
    &\Exp\left[\hat{\hbf}_{{\rm e}mk}\hat{\hbf}^{\herms}_{{\rm e}mk'}\right] = \notag\\
    &P_p {\Rbf}_{{\rm e}mk} \left( \vect{\psi}_k\herm \kron \Ibf_{N_{RF}} \right) \left(\vect{\Psi}_{\rm e}\Rbf_{{\rm e}m} \vect{\Psi}_{\rm e}\herm + \Cbf_{{\rm e}zm} \right)^{-1}\Exp\left[ \ybf_{{\rm e}m}  \ybf_{{\rm e}m'}^{\herms}\right] \notag\\
    &\times
    \left(\vect{\Psi}_{\rm e}\Rbf_{{\rm e}m'} \vect{\Psi}_{\rm e}\herm + \Cbf_{{\rm e}zm} \right)^{-1} \left( \vect{\psi}_{k'} \kron \Ibf_{N_{RF}} \right){\Rbf}_{{\rm e}m'k'}\notag\\
    &=\delta(m-m') P_p {\Rbf}_{{\rm e}mk} \left( \vect{\psi}_k\herm \kron \Ibf_{N_{RF}} \right) \left(\vect{\Psi}_{\rm e}\Rbf_{{\rm e}m} \vect{\Psi}_{\rm e}\herm +\Cbf_{{\rm e}zm} \right)^{-1} 
    \notag\\
    &\times \left( \vect{\psi}_{k'} \kron \Ibf_{N_{RF}} \right){\Rbf}_{{\rm e}mk'}.
\end{align}
On the other hand, the cross-covariance matrix of the estimated effective channels between CPU and users $k$ and $k'$ is
\begin{align*}
    \Exp\left[\hat{\hbf}_{{\rm e}k}\hat{\hbf}^{\herms}_{{\rm e}k}\right] = \Exp\begin{bmatrix} \hat{\hbf}_{{\rm e}1k} \\ \vdots\\ \hat{\hbf}_{{\rm e}Mk} \end{bmatrix} \begin{bmatrix} \hat{\hbf}_{{\rm e}1k}^{\herms} , \ldots, \hat{\hbf}_{{\rm e}Mk}^{\herms} \end{bmatrix}.
\end{align*}
Using $\Exp\left[\hat{\hbf}_{{\rm e}mk}\hat{\hbf}^{\herms}_{{\rm e}mk'}\right]$ from \eqref{eq:expmkmkp} completes the proof.
\end{proof}

In this paper, we base all our analysis and proofs on the assumption that all of the matrices are full rank. However, based on the values of matrix $\Wbf_m$ the matrix $\vect{\Psi}_{\rm e}\Rbf_{{\rm e}m} \vect{\Psi}_{\rm e}\herm + \Cbf_{{\rm e}mz}$ might not be invertible, if so we use Moore-Penrose pseudo-inverse, instead of the inverse. Moore-Penrose pseudo-inverse is known to be optimal for LMMSE estimation \cite[Theorem 3.2.3]{kailath2000linear}. For the rest of the paper, unless necessary, we remove the sub-index ${\rm e}$ to avoid notation clutter.

\section{Uplink Transmission}
\label{sec:up}

After channel estimation of the effective channels, we have the uplink data transmission. In this section, we discuss SINR expression for uplink transmission along with random matrix approximations of the SINR expression based on large scale statistics of the system, and users' power optimization. 

\subsection{Achievable Rate SINR Expression}
Following \cite{yang2013capacity,ngo2017cell}, given user $k$ has uplink SINR of SINR$^u_k$, it can achieve the uplink rate of 
\begin{align}
    R_k^u = \frac{\tau_c- \tau_p}{2\tau_c} \log_2(1+{\rm SINR}^u_k).
\end{align}
To calculate ${\rm SINR}^u_k$, consider user $k$ received signal at AP $m$
\begin{align}
    \ybf_{mk} = \sum_{k}\sqrt{\Pup_{k}}\hk s_k + \zbf_m,
\end{align}
where $\Pup_k$ is the maximum average uplink transmit power for user $k$, $s_k$ is the transmit symbol of user $k$ with $\Exp\left[|s_k|^2\right] = 1$, $\zbf_m=  \Wbf_m^{\herms}\xbf,~\xbf\sim \CN\left(\zero, \Ibf_N\right)$ is the noise vector at AP $m$, and $\hbf_{mk}$ and $\Wbf_m$ are the effective channel after analog beamforming between AP $m$ and user $k$ and analog beamforming matrix at AP $m$, respectively.

To decode the transmit signals, we perfom centralized decoding and process all the received signals from the APs jointly at the CPU. Using the combining vector $\vbf_k$ to decode $s_k$, the received decoded signal for user $k$ can be written as 
\begin{align*}
\begin{aligned}
    &\hat{s}_k = \\
    &\vbf_k\herm\left(\sqrt{\Pup_k} \hhk s_k + \sum_{k'\neq k} \sqrt{\Pup_{k'}} \hhkp s_{k'}  + \sum_{k'} \sqrt{\Pup_{k'}} \tilde{\hbf}_{{k'}} s_{k'}  +  \zbf \right),
    \end{aligned}
\end{align*}
where $\tilde{\hbf}_{{k}} = \hk- \hhk$ is the estimation error of the user $k$ effective channel to the CPU and $\zbf= \left[\zbf_1^{\trans},\ldots \zbf_M^{\trans} \right]\tran$. As a result, the $k^{\rm th }$ user SINR becomes
\begin{align}
\label{eq:Dval}
    &\SINRup_k  = \frac{\Pup_k|\vbf_k \herm \hhk |^2}{\vbf_k \herm\left(\sum_{k'\neq k} \Pup_{k'} \hhkp\hhkp\herm  + \Dbf \right) \vbf_k},\\
    &\Dbf = \sum_{k} \Pup_k  \left(\Rbf_{k} - \Cbf_{{k}}\right) + \Cbf_{z},
\end{align}
where $ \Cbf_{z} = \diag\{\Wbf_m^{\herms}\Wbf_m\}_{m \in [M]}$. Following generalized Rayleigh quotient result \cite[Lemma B.10]{bjornson2017massive}, to maximize the $\SINRup_k$, we need to use the MMSE combiner which is
\begin{align}
\label{eq:SINRvec}
    &\vbf_k = \vect{\Omega}^{-1}\hhk, \quad 
    \vect{\Omega} = \sum_{k'} \Pup_{k'} \hhkp\hhkp\herm  + \Dbf,
\end{align}
and the maximum SINR for user $k$ becomes 
\begin{align}
\label{eq:upSINRexp} 
        &\SINRup_k  = \Pup_k \hhk\herm \vect{\Omega}_k^{-1}\hhk,\quad 
        \vect{\Omega}_k = \sum_{k'\neq k} \Pup_{k'} \hhkp\hhkp\herm  + \Dbf.
\end{align}

Next, we provide random matrix approximations of the SINR expression in \eqref{eq:upSINRexp}. 

\subsection{Random Matrix Approximation}
Finding an accurate approximation of SINR which is only based on large scale statistics of the system is of great importance as it can be used to evaluate the system's performance and optimize its  parameters without the knowledge of the channel realizations. Next theorem, provides two random matrix approximations for the SINR expression in \eqref{eq:upSINRexp} which only depend on large scale statistics.

\begin{Theorem}
\label{thm:rmup}
The SINR expression in \eqref{eq:upSINRexp} can be approximated as follows. 

First approximation:
\begin{align}
\label{eq:ur1}
\begin{aligned}
   \SINRup_k \approx &\Pup_k  \tr~\! \Ckk\Dbf^{-1}\\
   &-\Pup_k\sum_{k'\neq k} \frac{ \Pup_{k'} \tr~\!  \Ckpkp \Dbf^{-1} \Ckk \Dbf^{-1} }{1+ \Pup_{k'}\tr~\! \Ckpkp\Dbf^{-1}}.
   \end{aligned}
\end{align}

Second approximation:
\begin{align}
\label{eq:ur2}
    & \SINRup_k \approx \frac{\Pup_k}{M\Nrf}\tr~\!\Ckk \Tbf,
\end{align}
where 
\begin{subequations}
\label{eq:Tval}
\begin{align}
    &\Tbf =\left(\frac{\Dbf}{MN_{RF}} +\frac{1}{MN_{RF}} \sum_{k} \frac{\Pup_k}{1+e_k}\Ckk \right)^{-1},\\
    & \text{with}~e_k = \lim_{t\rightarrow \infty} e_k(t),~\text{for}~e_k(0) = 1,~\text{for all}~ k\\
    &\begin{aligned}
    &e_k(t+1) = \frac{\Pup_k}{MN_{Rf}}\times \\
    &\tr~\!\Ckk  \left(\frac{\Dbf}{MN_{RF}} +\frac{1}{MN_{RF}} \sum_{k} \frac{\Pup_k}{1+e_k(t)}\Ckk \right)^{-1},
    \end{aligned}
\end{align}
\end{subequations}
\end{Theorem}
\begin{proof}
The proof is provided in Appendix~\ref{app:rmup}.
\end{proof}
As we will see in Sec.~\ref{sec:sim}, the first approximation is only accurate when the total number of RF-chains is greater than the number of users $M\Nrf > K$. But, the second approximation is accurate for any values of $K$ and $M\Nrf$. However, it is more computationally complex.

\subsection{Power Optimization}
\label{subsec:popup}
Next, we discuss the user power optimization for maximizing the minimum SINR among the users which corresponds to maximizing the worst rate of the system.
To optimization the transmit powers, we make use of the result in \cite{hong2014unified}. More specifically, in our next theorem, we show that the SINR expressions in \eqref{eq:upSINRexp} and 
\eqref{eq:ur2} follow the definition of the \textit{competitive utility functions} \cite[Assumption~1]{hong2014unified} and the power constraints are \textit{monotonic constraints} \cite[Assumption~2]{hong2014unified}. Therefore, we can repurpose \cite[Algorithm~1]{hong2014unified} as shown in Alg.~\ref{Alg:pOpt} to find optimal users' transmit powers. This algorithm is centralized in the sense that the CPU performs the optimization to find the optimal power allocation for the users and then sends the power values to the users.
\begin{Theorem}
\label{thm:pOptup}
Using Alg.~\ref{Alg:pOpt}, the transmit power of the users can be optimized to achieve optimal max-min SINR among the users for the SINR functions in \eqref{eq:upSINRexp} and 
\eqref{eq:ur2}.
\end{Theorem}
\begin{proof}
The proof is provided in Appendix~\ref{app:pOptup}.
\end{proof} 

\section{Downlink Transmission}
\label{sec:dl}

After channel estimation and uplink data transmission, we have the downlink transmission. In this section, we discuss SINR expression for downlink transmission along with random matrix approximations of the SINR expression based on large scale statistics of the system. 
\begin{algorithm2e}[t]
\caption{Uplink user power optimization}
\label{Alg:pOpt}
\textbf{Initialize:} $\vect{\Pup}(0)= [\Pup_1(0), \Pup_2(0), \ldots, \Pup_K(0)]\tran > 0$\;
\While{not converged}{
$\Pup_k(t+1) = \frac{\Pup_k(t)}{\SINRup_k\left(\vect{\Pup}(t)\right)}$\tcp*{\!Update power vector $\vect{\Pup}(t+1)$}
$\vect{\Pup}(t+1) = \frac{\vect{\Pup}(t+1)}{\max\{\vect{\Pup}(t+1)\}}$\tcp*{Scale power vector $\vect{\Pup}(t+1)$}
}
\textbf{Output:} $\vect{\Pup}(t+1)$
\end{algorithm2e}
\subsection{Achievable Rate and SINR Expression}

Following \cite{yang2013capacity,ngo2017cell}, given user $k$ has downlink SINR of SINR$^d_k$, it can achieve the downlink rate of 
\begin{align}
    R_k^d = \frac{\tau_c- \tau_p}{2\tau_c} \log_22(1+{\rm SINR}^d_k).
\end{align}

Next, we calculate the users' SINRs. Define $\Pdl_{mk}$ as the transmit power used at AP $m$ for downlink data transmission to user $k$ and $\vbf^d_{mk}$ as the encoder vector at AP $m$ used for user $k$. The received signal at user $k$ is
\begin{align}
    \begin{aligned}
        \ybf_k &= \sum_M \sqrt{\Pdl_{mk}} \Exp[\hbf^{\herms}_{mk} \vbf^d_{mk}]s_k \\
        &+ \sum_M \sqrt{\Pdl_{mk}}\left( \hbf^{\herms}_{mk} \vbf^d_{mk}- \Exp[\hbf^{\herms}_{mk} \vbf^d_{mk}]\right)s_k \\
        &+ \sum_M \sum_{k'\neq k} \sqrt{\Pdl_{mk'}} \hbf^{\herms}_{mk'} \vbf^d_{mk'} s_{k'} + \zbf_k,
    \end{aligned}
\end{align}
where $s_k$ is the transmit signal to user $k$ and $\zbf_k$ is the noise vector whose elements are i.i.d. Gaussian $\CN(0,1)$.

We assume that the CPU uses the RZF combiner which is
\begin{align}
\label{eq:dcomb}
    \hat{\vbf}_k^d = \vect{\Omega}_{\rm RZF}^{-1}\hhk, \quad 
    \vect{\Omega}_{\rm RZF} = \sum_{k'}  \hhkp\hhkp\herm  + \rho\Ibf,
\end{align}
where $\rho$ is the regularization factor and $\hat{\vbf}_k^d =  [\hat{\vbf}^d_{1k}, \ldots \hat{\vbf}^d_{MK}]$ with $\hat{\vbf}^d_{mk}$ denoting the combiner vector used at AP $m$. 
Note that we cannot use $\hat{\vbf}^d_k$ directly as the transmit power constraint at the APs is not satisfied. Also, if scaling of the combiner vector for each AP is different then the resulting combiner vector will not have the interference cancellation property of the RZF. Therefore, here, we scale $\hat{\vbf}_k^d $ such that the AP with maximum norm combiner vector $\hat{\vbf}_{mk}^d $ uses transmit power $P_k^d$ for user $k$. We have
\begin{align}
    \vbf^d_{k} = \frac{\hat{\vbf}^d_{k}}{\max_m \sqrt{ \Exp\{|\hat{\vbf}^d_{mk}|^2\}}}.
\end{align}
As a result, for $\SINRdl_k$, we have
\begin{align} 
\label{eq:dlSINRexp}
\begin{aligned}
    &\SINRdl_k =\\
    &\frac{ P^d_k|  \Exp \{\hbf_k\herm \vbf_k^d \} |^2  }{ \sum_{k' \neq  k}{P^d_{k'}\mathbb{E}}\left\{| \hbf_k\herm  \vbf_{k'}^d |^2\right\} +  P^d_k{\mathbb{V}}\left\{ \hbf_k\herm \vbf_k^d \right\} + 1},
    \end{aligned}
\end{align}

\begin{figure*}[t]
\centering
\begin{subfigure}{0.32\linewidth}
    \centering
    \includegraphics[width = 0.9\textwidth]{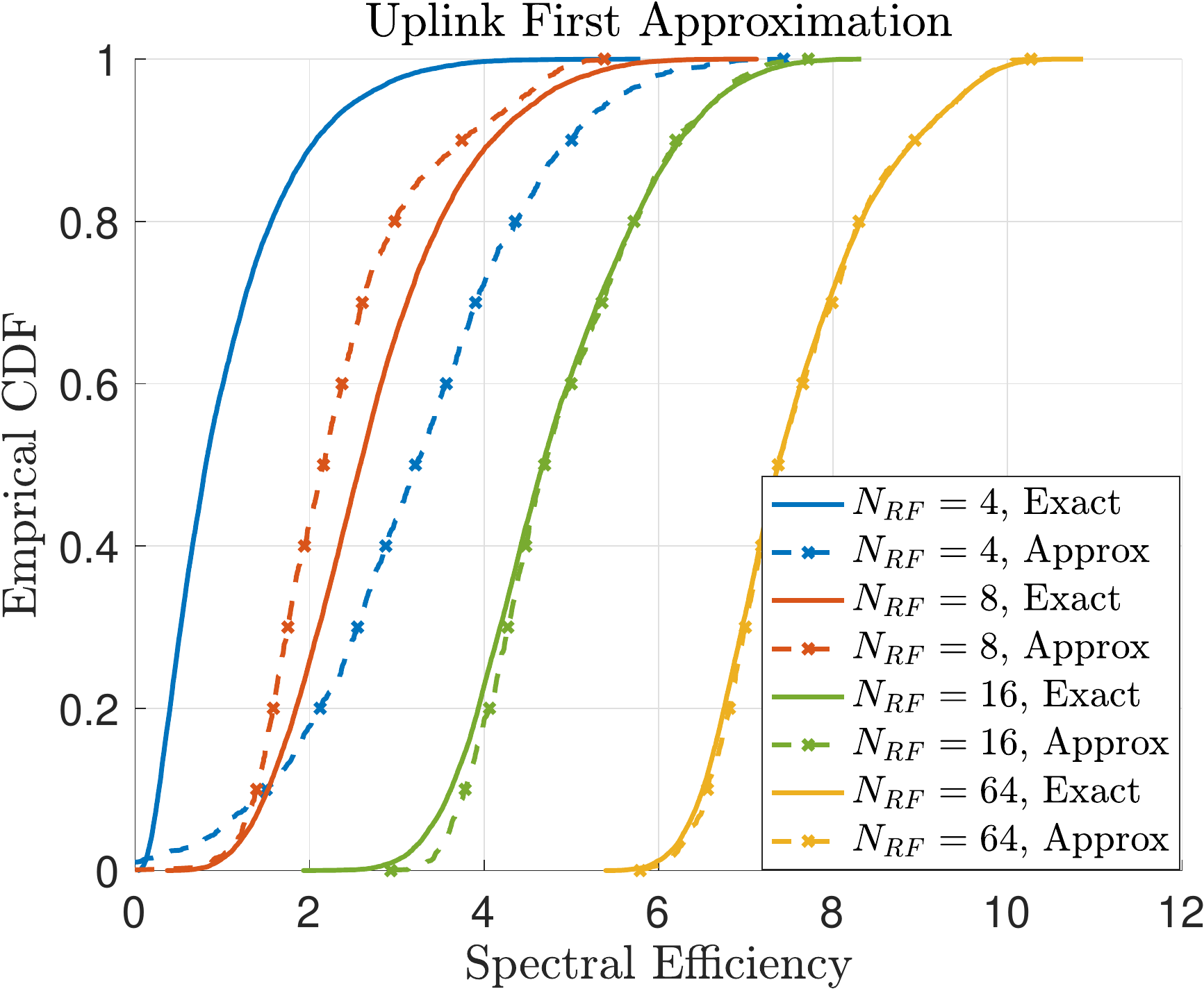}
    \caption{}
    \label{fig:upA11}
\end{subfigure}
\begin{subfigure}{0.32\linewidth}
    \includegraphics[width = 0.9\textwidth]{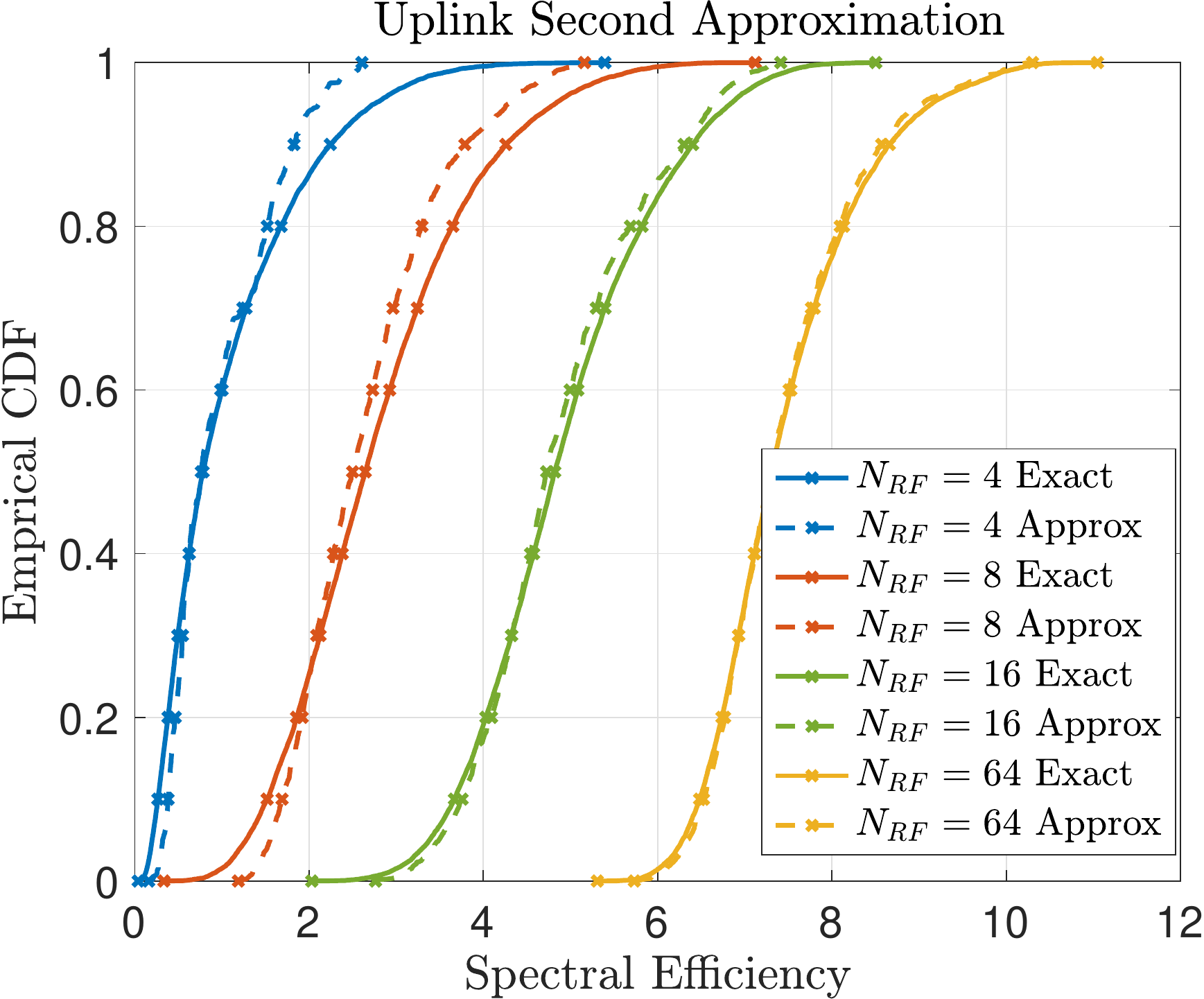}
    \caption{} 
    \label{fig:upA2}
\end{subfigure}
\begin{subfigure}{0.32\linewidth}
    \centering
    \includegraphics[width = 0.9\linewidth]{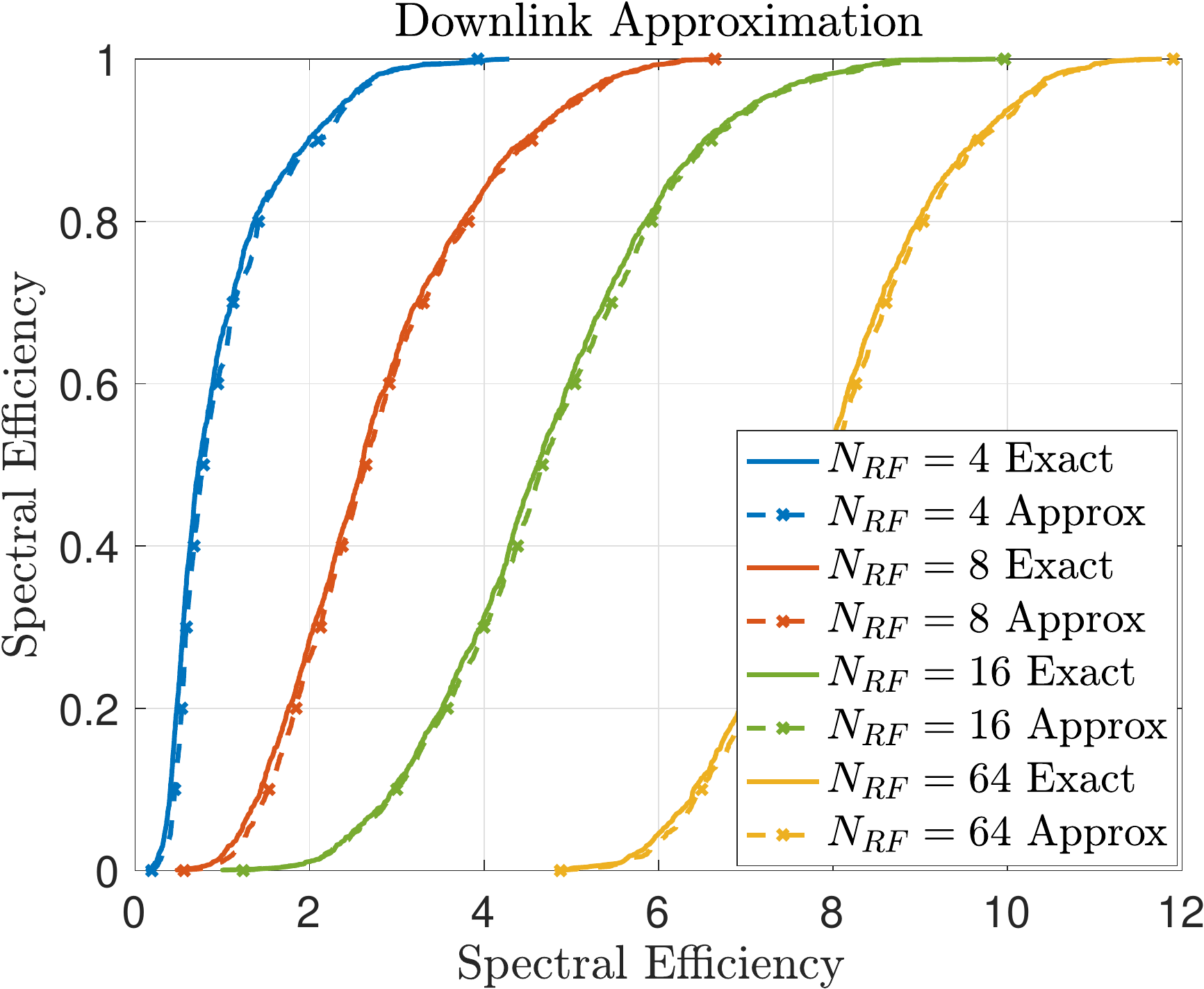}
    \caption{}
    \label{fig:dlA}
\end{subfigure}
\vspace*{-0.1cm}
\caption{Empirical CDF of spectral efficiency based on exact uplink SINR expression \eqref{eq:upSINRexp} and (a) approximation \eqref{eq:ur1}, (b)  approximation \eqref{eq:ur2}. (c) Empirical CDF of spectral efficiency based on exact downlink SINR expression \eqref{eq:dlSINRexp} and its approximation in \eqref{eq:apdlSINR}.}
\vspace*{-0.4cm}
\label{fig:upA1}
\end{figure*} 

\subsection{Random Matrix Approximation}
Next theorem, provides a random matrix approximations for $\SINRdl_k$ in \eqref{eq:dlSINRexp} which only depend on large scale statistics. 
\begin{Theorem}
\label{thm:rmdl}
The SINR expression in \eqref{eq:dlSINRexp} can be approximated as follows.
\begin{align}
    \label{eq:apdlSINR}
     \SINRdl_k\approx \frac{P^d_k\alpha_k}{\sum_{k'\neq k}P^d_{k'}\beta_{kk'}+1},
\end{align}
where 
\begin{align}
    &\alpha_k = \left(\frac{1}{\nu_k} \frac{\tr~\! \Ckk \Sbf }{M\Nrf+  \tr~\! \Ckk\Sbf}\right)^2 \\
    &\begin{aligned}
    &\beta_{kk'} =     \frac{1}{\nu_{k'}^2(M\Nrf + \tr~\! \Ckpkp \Sbf)^2}\times\Big[ \tr~\!\Rbf_k\Sbf'_{k'} \\
    &- 2{\rm Re}\left\{\frac{\tr~\!\Ckk \Sbf\, \tr~\! \Ckk\Tbf'_{k'}}{M\Nrf + \tr~\! \Ckk \Sbf}\right\}\\
    &+(M\Nrf + \tr~\! \Ckk \Sbf)^2(\tr~\! \Ckk \Sbf )^2 \,\tr~\! \Ckk \Sbf'_{k'}\Big ]
    \end{aligned}\\
     &\nu_k^2 = \frac{\max_{m\in[M]}\tr~\! \Ckk[m]\Sbf'[m]}{\left(M\Nrf + \tr~\! \Ckk\Sbf\right)^2},\quad\text{for all $k$}
\end{align}
where $\Ckk[m]$ and $\Sbf'[m]$ are the $m^{\rm th}$ matrices of size $\Nrf\times \Nrf$ located on the diagonal of $\Ckk$ and $\Sbf'$, respectively,
\begin{subequations}
\label{eq:sbf}
\begin{align}
    &\Sbf =\left(\frac{\rho\Ibf}{M\Nrf} +\frac{1}{M\Nrf} \sum_{k} \frac{1}{1+e_{k}}\Ckk \right)^{-1},\\
    & \text{with}~e_k = \lim_{t\rightarrow \infty} e_k(t),~\text{for}~e_k(0) = 1,~\text{for all}~ k \\
    &\begin{aligned}
    &e_k(t+1) = \frac{1}{MN_{Rf}}\times \\
    &\tr~\!\Ckk  \left(\frac{\rho\Ibf}{MN_{RF}} +\frac{1}{MN_{RF}} \sum_{k} \frac{1}{1+e_k(t)}\Ckk \right)^{-1},
    \end{aligned}
\end{align}
\end{subequations}
and
\begin{subequations}
\label{eq:sbfp}
\begin{gather}
    \Sbf' = \Sbf^2 + \Sbf\frac{1}{\Nrf}\sum_{k=1}^K \frac{\Ckk e'_{k}}{(1+ e_{k})^2}\Sbf,\\
    e'_{k} = \left(\Ibf -\Jbf\right)^{-1}\vbf,\\
    [\vbf]_{k} = \frac{1}{\Nrf}\tr~\!\left( \Ckk\Sbf^2\right),     [\Jbf]_{kk'}= \frac{\tr~\!\left(\Ckk \Sbf\Ckpkp\Sbf\right)}{\Nrf^2(1+ e_{k'})^2},
    \end{gather}
\end{subequations}
\vspace{-1cm}
\begin{subequations}
\label{eq:sbfp1}
    \begin{gather}
    \mathbf{S}'_{k'} = \Sbf\Ckpkp\Sbf + \Sbf\frac{1}{M\Nrf}\sum_{k=1}^K \frac{\Ckpkp e''_k}{(1+ e_k)^2}\Sbf,\\
    e''_k = \left(\Ibf -\Jbf\right)^{-1}\vbf_{k'},\\
    [\vbf_{k'}]_{k} = \frac{1}{M\Nrf}\tr~\!\left(\Ckk\Sbf\Ckpkp\Sbf\right),
\end{gather}
\end{subequations}
\end{Theorem}
\begin{proof}
The proof is provided in Appendix~\ref{app:rmdl}.
\end{proof}

Optimizing the total transmit power of each user to maximize the minimum downlink SINR among the users is an interesting venue left for future publication due to space constraint.  

\begin{table}[t]
    \renewcommand{\arraystretch}{1.4}
    \begin{center}
        \caption{Simulation Setup and Parameters.}~\label{tbl:1}
        \fontsize{8}{8}\selectfont
        \begin{tabular}{|>{\centering\arraybackslash}m{6cm}|>{\centering\arraybackslash}m{1.5cm}|>{\centering\arraybackslash}m{0.6cm}|}\hline
            \textbf{Parameter}       &    \textbf{Value}     \\ \hline
            $f_c$ (Carrier frequency)        &    1.9 GHz            \\ \hline
            BW (Bandwidth)               &    20 MHz             \\ \hline
            Noise figure             &    9 dB               \\ \hline
            $\tau_c$ (Length of coherence Interval) & 200 \\ \hline
            $\tau_p$ (Length of pilot sequence)  & 16 \\ \hline

            $P_{u}$ (UL maximum transmit power per data symbol)  &  20 mW             \\ \hline
            $P_{p}$ (UL maximum transmit power per pilot symbol) & 20 mW             \\ \hline
            $P_{d}$ (DL maximum transmit power per user) & 200 mW             \\ \hline
$K$ (Number of simultaneous users) & $16$\\ \hline 
            Cell area & $100\! \times\! 100~\!$m$^2$ \\ \hline
        \end{tabular}
    \end{center}
\end{table}
\section{Simulations and Numerical Evaluations}
\label{sec:sim}

\begin{figure}[t]
    \centering
    \includegraphics[width = 0.8\linewidth]{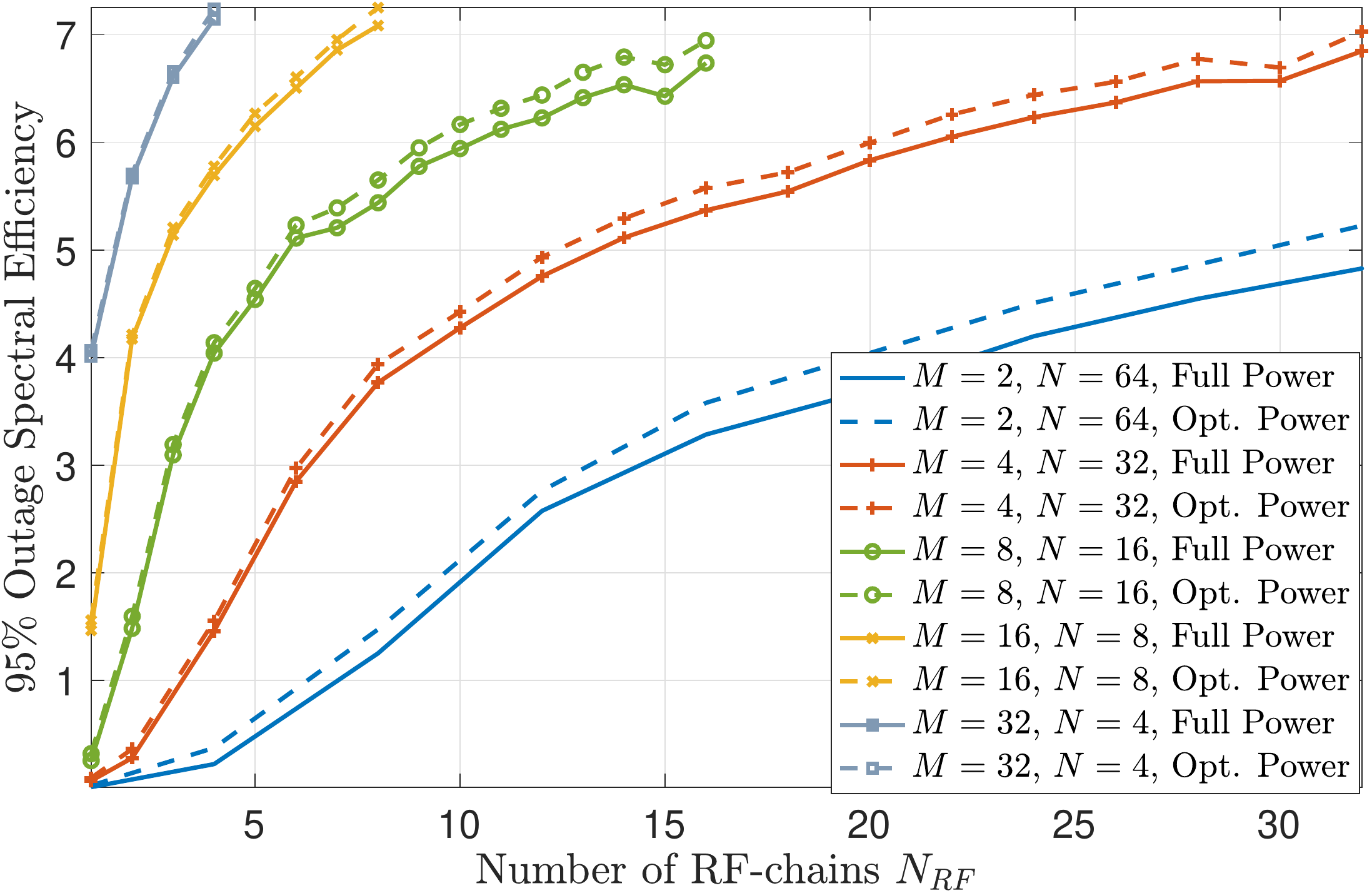}
    \caption{$95\%$ outage spectral efficiency for different number of APs, antennas, and RF-chains.}
    \label{fig:powopt}
\end{figure}

For our simulations, we use the parameters presented in Table~\ref{tbl:1}. Even though orthogonal pilots can be used since pilot length $= 16 = K$, we use random pilots in our simulations.
First, we investigate the accuracy of the provided approximations in Thm.~\ref{thm:rmup} and Thm.~\ref{thm:rmdl} for the uplink and downlink SINRs, respectively, then we investigate the impact of number of RF-chains and power optimization on the system's performance for uplink scenario.

\subsection{SINR Approximations}
The empirical CDF of uplink spectral efficiency based on exact SINR \eqref{eq:upSINRexp} and its first and second approximations  \eqref{eq:ur1}, \eqref{eq:ur2} are plotted for different number of users and RF-chains in Fig.~\ref{fig:upA11} and Fig.~\ref{fig:upA2}, respectively. We observe that the second approximation is good almost everywhere while the first approximation is only good when $M\Nrf>K$. However, the first approximation has smaller computation complexity compared to the second one.

The empirical CDF of uplink spectral efficiency based on exact downlink SINR in \eqref{eq:dlSINRexp} and its approximation in Thm.~\ref{thm:rmdl} is provided in Fig.~\ref{fig:dlA}. We observe that the approximation accurately estimates the exact CDF everywhere.

\subsection{Uplink Power optimization and Impact of number of RF-chains}

The $95\%$ outage spectral efficiency for different number of APs, antennas, and RF-chains are plotted in Fig.~\ref{fig:powopt}. Here, the solid lines are calculated based on the exact expression \eqref{eq:upSINRexp} when all users transmit with full power and the dashed lines are when the transmit power of the users are optimized using the  Alg.~\ref{Alg:pOpt} using the SINR approximation in \eqref{eq:ur2}. We observe that maximizing the minimum SINR does not improve the performance of the system visibly. Moreover, Fig.~\ref{fig:powopt} gives an important insight on the trade off between the number of RF-chains and system performance for different values of $M$ and $N$ considering the discussed BF technique in Sec.~\ref{subsec:ABF}.

\section{Conclusion}
In this paper, a Cell-Free mMIMO system with hybrid APs is considered ,where each AP uses a hybrid beamfoming scheme in which the analog beamformer is designed based on the large scale statistics of the channel, and the digital beamformer is designed based on the estimated effective channel between the APs and users. Using MMSE combiner for uplink and RZF precoder for downlink, SINR expressions for both uplink and downlink scenarios along with accurate random matrix approximations of the SINR expressions which only depend on large scale statistics of the system are provided. Furthermore, optimality achieving power optimization method that maximizes minimum SINR among the users for the uplink scenario is presented. Finally, through various simulations of practical systems the accuracy of the derived approximations is validated and the impact of power optimization, and number of RF-chains on the system's perfomance is investigated. 

\bibliography{main_ICC.bbl}

\begin{thebibliography}{10}
\providecommand{\url}[1]{#1}
\csname url@samestyle\endcsname
\providecommand{\newblock}{\relax}
\providecommand{\bibinfo}[2]{#2}
\providecommand{\BIBentrySTDinterwordspacing}{\spaceskip=0pt\relax}
\providecommand{\BIBentryALTinterwordstretchfactor}{4}
\providecommand{\BIBentryALTinterwordspacing}{\spaceskip=\fontdimen2\font plus
\BIBentryALTinterwordstretchfactor\fontdimen3\font minus
  \fontdimen4\font\relax}
\providecommand{\BIBforeignlanguage}[2]{{%
\expandafter\ifx\csname l@#1\endcsname\relax
\typeout{** WARNING: IEEEtran.bst: No hyphenation pattern has been}%
\typeout{** loaded for the language `#1'. Using the pattern for}%
\typeout{** the default language instead.}%
\else
\language=\csname l@#1\endcsname
\fi
#2}}
\providecommand{\BIBdecl}{\relax}
\BIBdecl

\bibitem{shamai2001enhancing}
S.~Shamai and B.~M. Zaidel, ``Enhancing the cellular downlink capacity via
  co-processing at the transmitting end,'' in \emph{IEEE VTS 53rd Vehicular
  Technology Conference, Spring 2001. Proceedings (Cat. No. 01CH37202)},
  vol.~3.\hskip 1em plus 0.5em minus 0.4em\relax IEEE, 2001, pp. 1745--1749.

\bibitem{bjornson2020scalable}
E.~Bj{\"o}rnson and L.~Sanguinetti, ``Scalable cell-free massive {MIMO}
  systems,'' \emph{IEEE Transactions on Communications}, vol.~68, no.~7, pp.
  4247--4261, 2020.

\bibitem{ngo2017cell}
H.~Q. Ngo \emph{et~al.}, ``Cell-free massive {MIMO} versus small cells,''
  \emph{IEEE Transactions on Wireless Communications}, vol.~16, no.~3, pp.
  1834--1850, 2017.

\bibitem{nayebi2017precoding}
E.~Nayebi \emph{et~al.}, ``Precoding and power optimization in cell-free
  massive {MIMO} systems,'' \emph{IEEE Transactions on Wireless
  Communications}, vol.~16, no.~7, pp. 4445--4459, 2017.

\bibitem{bjornson2019making}
E.~Bj{\"o}rnson and L.~Sanguinetti, ``Making cell-free massive {MIMO}
  competitive with {MMSE} processing and centralized implementation,''
  \emph{IEEE Transactions on Wireless Communications}, vol.~19, no.~1, pp.
  77--90, 2019.

\bibitem{femenias2019cell}
G.~Femenias and F.~Riera-Palou, ``Cell-free millimeter-wave massive {MIMO}
  systems with limited fronthaul capacity,'' \emph{IEEE Access}, vol.~7, pp.
  44\,596--44\,612, 2019.

\bibitem{rangan2014millimeter}
S.~Rangan, T.~S. Rappaport, and E.~Erkip, ``Millimeter-wave cellular wireless
  networks: Potentials and challenges,'' \emph{Proceedings of the IEEE}, vol.
  102, no.~3, pp. 366--385, Mar. 2014.

\bibitem{el2014spatially}
O.~El~Ayach \emph{et~al.}, ``Spatially sparse precoding in millimeter wave
  {MIMO} systems,'' \emph{IEEE transactions on wireless communications},
  vol.~13, no.~3, pp. 1499--1513, 2014.

\bibitem{gao2016energy}
X.~Gao \emph{et~al.}, ``Energy-efficient hybrid analog and digital precoding
  for mmwave mimo systems with large antenna arrays,'' \emph{IEEE Journal on
  Selected Areas in Communications}, vol.~34, no.~4, pp. 998--1009, 2016.

\bibitem{yan2021scalable}
H.~Yan, A.~Ashikhmin, and H.~Yang, ``A scalable and energy efficient {IoT}
  system supported by cell-free massive {MIMO},'' \emph{IEEE Internet of Things
  Journal}, 2021.

\bibitem{park2017exploiting}
S.~Park \emph{et~al.}, ``Exploiting spatial channel covariance for hybrid
  precoding in massive {MIMO} systems,'' \emph{IEEE Transactions on Signal
  Processing}, vol.~65, no.~14, pp. 3818--3832, 2017.

\bibitem{zhu2016novel}
D.~Zhu, B.~Li, and P.~Liang, ``A novel hybrid beamforming algorithm with
  unified analog beamforming by subspace construction based on partial {CSI}
  for massive {MIMO-OFDM} systems,'' \emph{IEEE Transactions on
  Communications}, vol.~65, no.~2, pp. 594--607, 2016.

\bibitem{mai2018two}
R.~Mai, T.~Le-Ngoc, and D.~H. Nguyen, ``Two-timescale hybrid {RF}-baseband
  precoding with {MMSE-VP} for multi-user massive {MIMO} broadcast channels,''
  \emph{IEEE Transactions on Wireless Communications}, vol.~17, no.~7, pp.
  4462--4476, 2018.

\bibitem{kailath2000linear}
T.~Kailath, A.~H. Sayed, and B.~Hassibi, \emph{Linear estimation}.\hskip 1em
  plus 0.5em minus 0.4em\relax Prentice Hall, 2000.

\bibitem{yang2013capacity}
H.~Yang and T.~L. Marzetta, ``Capacity performance of multicell large-scale
  antenna systems,'' in \emph{2013 51st Annual Allerton Conference on
  Communication, Control, and Computing (Allerton)}.\hskip 1em plus 0.5em minus
  0.4em\relax IEEE, 2013, pp. 668--675.

\bibitem{bjornson2017massive}
E.~Bj{\"o}rnson, J.~Hoydis, and L.~Sanguinetti, ``Massive {MIMO} networks:
  Spectral, energy, and hardware efficiency,'' \emph{Foundations and Trends in
  Signal Processing}, vol.~11, no. 3-4, pp. 154--655, 2017.

\bibitem{hong2014unified}
Y.-W.~P. Hong \emph{et~al.}, ``A unified framework for wireless max-min utility
  optimization with general monotonic constraints,'' in \emph{IEEE INFOCOM
  2014-IEEE Conference on Computer Communications}.\hskip 1em plus 0.5em minus
  0.4em\relax IEEE, 2014, pp. 2076--2084.

\bibitem{silverstein1995empirical}
J.~W. Silverstein and Z.~Bai, ``On the empirical distribution of eigenvalues of
  a class of large dimensional random matrices,'' \emph{Journal of Multivariate
  analysis}, vol.~54, no.~2, pp. 175--192, 1995.

\bibitem{wagner2012large}
S.~Wagner \emph{et~al.}, ``Large system analysis of linear precoding in
  correlated {MISO} broadcast channels under limited feedback,'' \emph{IEEE
  transactions on information theory}, vol.~58, no.~7, pp. 4509--4537, 2012.

\bibitem{6415388}
J.~Hoydis, S.~ten Brink, and M.~Debbah, ``Massive {MIMO} in the {UL}/{DL} of
  cellular networks: How many antennas do we need?'' \emph{IEEE Journal on
  Selected Areas in Communications}, vol.~31, no.~2, pp. 160--171, 2013.

\end{thebibliography}

\clearpage
\appendices
\section{Useful Lemmas and Theorems}
For our derivations, we make use of the following results

\begin{lemma} \label{Lemma 1}(\textbf{Matrix Inversion Lemma})\cite{silverstein1995empirical}
    Let $\mathbf{U}$ be an $M \times M$ invertible matrix and $\mathbf{x} \in \mathbb{C}^{M\times 1}$, $c\in \mathbb{C}$ for which $\mathbf{U} + c\mathbf{xx}^{H}$ is invertible. Then
    \begin{equation}
    \mathbf{x}^{H}\left(\mathbf{U} + c\mathbf{xx}^{H}\right)^{-1} = \frac{\mathbf{x}^{H}\mathbf{U}^{-1}}{1 + c\mathbf{x}^{H}\mathbf{U}^{-1}\mathbf{x}}.
    \end{equation}
\end{lemma}

\begin{lemma}\cite[Lemma 4 and 5]{wagner2012large}\label{Lemma 2}
    Let $\mathbf{A} \in \mathbb{C}^{M\times M}$ and $\mathbf{x},\mathbf{y} \sim \mathcal{CN} (0, \frac{1}{M}\mathbf{I}_{M}).$ Assume that $\mathbf{A}$ has uniformly bounded spectral norm (with respect to $M$) and that $\mathbf{x}$ and $\mathbf{y}$ are mutually independent and independent of $\mathbf{A}$. Then,
    \begin{subequations}
    \begin{align}
    \mathbf{x}^{H}\mathbf{A}\mathbf{x} - \frac{1}{M}\text{tr}\mathbf{A}\xrightarrow[M \rightarrow \infty]{\text{a.s.}} 0,
    \end{align}
    \begin{align}
    \label{eq:x_A_y}
    \mathbf{x}^{H}\mathbf{A}\mathbf{y} \xrightarrow[M \rightarrow \infty]{\text{a.s.}} 0.
    \end{align}
    \end{subequations}
\end{lemma}

\begin{lemma}\cite[Lemma 6]{wagner2012large}\label{Lemma 3}
Let $\mathbf{A} \in \mathbb{C}^{M\times M}$, be deterministic with uniformly bounded spectral norm and $\mathbf{B} \in \mathbb{C}^{M\times M}$, be random Hermitian, with eigenvalues $\lambda_{1}^{\mathbf{B}_{M}}\leq \lambda_{2}^{\mathbf{B}_{M}} \cdots \leq \lambda_{M}^{\mathbf{B}_{M}}$ such that, with probability 1, there exist $\epsilon > 0$ for which $\lambda_{1}^{\mathbf{B}_{M}} > \epsilon$ for all large $M$. Then for $v \in \mathbb{C}^{M\times 1}$
\begin{equation}
\frac{1}{M}\tr~\!\mathbf{A}\mathbf{B}^{-1} - \frac{1}{M}\tr~\!\mathbf{A}\left(\mathbf{B} + \mathbf{vv}^{H}\right)^{-1} \xrightarrow[M \rightarrow \infty]{\text{a.s.}} 0
\end{equation} 
almost surely, where $\mathbf{B}^{-1}$ and $\left(\mathbf{B} + \mathbf{vv}^{H}\right)^{-1}$ exist with probability 1.   
\end{lemma}

\begin{theorem}\cite[Thm. 1]{wagner2012large}\label{Theorem 1}
Let $\mathbf{S} \in \mathbb{C}^{M\times M}$ and  $\mathbf{Q} \in \mathbb{C}^{M\times M}$ be Hermitian non-negative definite and let $\mathbf{X} \in \mathbb{C}^{M \times K}$ be random with independent column vectors $\mathbf{x}_{k} \sim \mathcal{CN} (0, \frac{1}{M}\mathbf{R}_{k})$ and $\mathbf{R}_{k},\, k = 1,..., K$ have uniformly bounded spectral norms (with respect to M). Then, for any $z > 0$, 
\begin{equation}
\begin{aligned}
    &\frac{1}{M}\tr~\!\mathbf{Q}\left(\mathbf{X}\mathbf{X}^{H} + \mathbf{S} + z\mathbf{I}\right)^{-1} - \frac{1}{M}\tr~\!\mathbf{Q}\mathbf{T}&\xrightarrow[M \rightarrow \infty]{\text{a.s.}} 0,
\end{aligned}
\end{equation}
where $\mathbf{T} \in \mathbb{C}^{M\times M}$ is given by
\begin{equation}
\mathbf{T} = \left(\frac{1}{M}\sum_{k=1}^{K}\frac{\mathbf{R}_{k}}{1 + e_{k}} + \mathbf{S} + z\mathbf{I}_{M}\right)^{-1}.
\end{equation}
Here, $e_{k} = \lim_{t\rightarrow \infty} e_{k}{(t)}$ and $e_{k}{(t)}$ is obtained by
\begin{equation}
e_{k}{(t)} = \frac{1}{M}\tr~\!\mathbf{R}_{k}\bigg(\frac{1}{M}\sum_{k'=1}^{K}\frac{\mathbf{R}_{k'}}{1 + e_{k'}{(t-1)}} + \mathbf{S}+ z\mathbf{I}_{M}\bigg)^{-1}.
\end{equation} 
where $t = 1, 2,...$ and $e_{k}^{(0)}(z) = 1/z$ for $k = 1, 2,..., K$. 
\end{theorem}

\begin{theorem}\cite[Thm. 1]{wagner2012large}\label{Theorem 2}
Let be ${\bf \Theta}$ Hermitian nonnegative definite with uniformly bounded spectral norm (with respect to N). Under the conditions of Theorem \ref{Theorem 1},
\begin{align}
\begin{aligned}
     &{1\over M} {\rm tr} \mathbf{Q}\left(\mathbf{X}\mathbf{X}^{H} + \mathbf{S} + z\mathbf{I}_{M}\right)^{-1}\!\!\!{\bf \Theta}\!\left(\mathbf{X}\mathbf{X}^{H} + \mathbf{S} + z\mathbf{I}_{M}\right)^{-1} \\
     &-{1\over M} {\tr~\!} {\bf \mathbf{Q} T}^{\prime} \xrightarrow[M \rightarrow \infty]{\text{a.s.}} 0.
\end{aligned}
\end{align}
where
\begin{align}
    {\bf T}^{\prime}={\bf T}{\bf \Theta} {\bf T}+{\bf T}{1\over M}\sum_{k=1}^{K}{{\bf R}_{k}e_{k}^{\prime}\over (1+e_{k})^{2}}{\bf T},
\end{align}
where $e_k$ is the same as in Theorem~\ref{Theorem 1} and 
\begin{subequations}
\begin{gather}
    \ebf' = (\Ibf_K- \Jbf)^{-1} \vbf,\\
    v_k = \frac{1}{M} {\rm tr} {\bf R}_{k}{\bf T}{\bf \Theta} {\bf T}, \\
    J_{kk'} = \frac{1}{M^2(1+ e_{k'})^2} {\rm tr} {\bf R}_{k}{\bf T}{\bf R}_{k'}{\bf T}
\end{gather}
\end{subequations}
\end{theorem}

\section{ Proof of Thm.~\ref{thm:rmup}}
\label{app:rmup}

For our approximations, we consider that the elements of matrix $\Ckkp$ for $k\neq k'$ are close to zero. Note that this is not generally true. Two possible cases where this is correct is when orthogonal pilots or random pilots with large length are used. However, through numerical evaluations, we observe that even for small pilot length this value is small and does not impact the approximation. 

\textbf{First Approximation:}
Consider the SINR expression in \eqref{eq:upSINRexp}, using Woodbury inversion lemma, we have 

\begin{align}
\label{eq:fa1}
\begin{aligned}
    &\SINRup_k =  \Pup_k \hhk\herm \Dbf^{-1} \hhk - \\
    &\Pup_k \hhk\herm \Dbf^{-1} \hat{\Hbf}_{/k}  \left( \Ibf_K + \hat{\Hbf}_{/k} \herm \Dbf^{-1} \hat{\Hbf}_{/k}  \right)^{-1} \hat{\Hbf}_{/k} \herm \Dbf^{-1} \hhk,
    \end{aligned}
\end{align}
where $\hat{\Hbf}  = [\sqrt{\Pup_1} \hat{\hbf}_{1}, \sqrt{\Pup_2}\hat{\hbf}_{2},\ldots,\sqrt{\Pup_K}\hat{\hbf}_{K}]$ and $\hat{\Hbf}_{/k}$ means that the column $k$ is removed. Let us consider the term
\begin{align}
     &\begin{aligned}
    &\hat{\Hbf}\herm \Dbf^{-1} \hat{\Hbf}= \\
    &\begin{bmatrix}\sqrt{\Pup_1}\hat{\hbf}_{1}\herm \\ \sqrt{\Pup_2}\hat{\hbf}_{2}\herm\\ \vdots\\ \sqrt{\Pup_K}\hat{\hbf}_{K}\herm \end{bmatrix} \Dbf^{-1} \begin{bmatrix} \sqrt{\Pup_1}\hat{\hbf}_{1}, \sqrt{\Pup_2}\hat{\hbf}_{2},\ldots,\sqrt{\Pup_K}\hat{\hbf}_{K}\end{bmatrix}
     \end{aligned}\\
     &\Rightarrow \left[ \hat{\Hbf}\herm \Dbf^{-1} \hat{\Hbf}\right]_{k,k'} =\sqrt{\Pup_{k}\Pup_{k'}}\, \hhk\herm \Dbf^{-1} \hhkp.
\end{align}
For $k = k'$, we have 
\begin{align}
\label{eq:fa2}
\begin{aligned}
    \hhk\herm \Dbf^{-1} \hhk &\xrightarrow[\text{ Distribution}]{\text{In}} \xbf \herm \Ckk^{0.5}\Dbf^{-1} \Ckk^{0.5} \xbf, \quad \xbf \sim \CN\left(\zero, \Ibf_{M\Nrf}\right) \\
    &\xrightarrow[]{\text{Lemma \ref{Lemma 2}}} \tr~\! \Ckk^{0.5}\Dbf^{-1} \Ckk^{0.5} = \tr~\! \Ckk\Dbf^{-1}.
    \end{aligned}
\end{align}
For $k \neq k'$, note that we can write
\begin{gather}
\begin{bmatrix}
    \hhk\\
    \hhkp 
\end{bmatrix} = \begin{bmatrix}
    \Abf_k, &\Bbf_k\\
    \Abf_{k'}, &\Bbf_{k'}
\end{bmatrix} \begin{bmatrix}
    \xbf\\
    \ybf 
\end{bmatrix},  \quad \xbf,\ybf \sim \CN\left(\zero, \Ibf_{M\Nrf}\right) \\
\Rightarrow  \begin{bmatrix}
    \Abf_k, &\Bbf_k\\
    \Abf_{k'}, &\Bbf_{k'}
\end{bmatrix}  =  \begin{bmatrix}
    \Ckk, &\Ckkp\\
    \Ckpk, &\Ckpkp
\end{bmatrix}^{0.5}
\end{gather}
we have 
\begin{align}
\label{eq:fa22}
    \hhk\herm \Dbf^{-1} \hhkp &\xrightarrow[\text{ Distribution}]{\text{In}} \left(\Abf_k\xbf\! +\! \Bbf_k\ybf\right) \!\herm \Dbf^{-1}\!\! \left(\Abf_{k'}\xbf\! +\! \Bbf_{k'}\ybf\right) \notag\\
    &\xrightarrow[]{\text{Lemma \ref{Lemma 2}}} \tr~\! \left(\Abf_{k} \Abf_{k'}\herm + \Bbf_{k} \Bbf_{k'}\herm\right) \Dbf^{-1}\notag\\
    &=  \tr~\! \Ckkp\Dbf^{-1} \approx 0,
\end{align}
Based on \eqref{eq:fa2} and \eqref{eq:fa22}, we have 
\begin{align}
     \left( \Ibf_K +\hat{\Hbf}_{/k}  \herm \Dbf^{-1} \hat{\Hbf}_{/k}  \right)^{-1}  \approx \diag\left\{ \frac{1}{1+\Pup_{k'}\tr~\! \Ckpkp\Dbf^{-1} }\right\}_{k'\neq k}.
\end{align}
Substituting this in \eqref{eq:fa1}, we get 
\begin{align}
\label{eq:fa3}
    \SINRup_k =  \Pup_k \hhk\herm \Dbf^{-1} \hhk - \Pup_k \sum_{k'\neq k} \frac{ \Pup_{k'} \left| \hhk\herm \Dbf^{-1} \hhkp \right|^2}{1+\Pup_{k'} \tr~\! \Ckpkp\Dbf^{-1}}.
\end{align}
Based on \eqref{eq:fa2}, already have an approximation for the first term. Next, we approximate the nominator. 
\begin{align}
\label{eq:fa4}
&\left| \hhk\herm \Dbf^{-1} \hhkp \right|^2 =  \hhk\herm \Dbf^{-1} \hhkp \hhkp \herm \Dbf^{-1} \hhk \notag\\
    &\xrightarrow[\text{ Distribution}]{\text{In}} \xbf \herm \Ckk^{0.5}\Dbf^{-1} \hhkp \hhkp \herm \Dbf^{-1} \Ckk^{0.5}\xbf,
    \notag\\
     &\xrightarrow[]{\text{Lemma \ref{Lemma 2}}} \tr~\! \Ckk^{0.5} \Dbf^{-1}   \hhkp \hhkp \herm \Dbf^{-1} \Ckk^{0.5} = \tr~\!  \hhkp \herm  \Dbf^{-1} \Ckk \Dbf^{-1} \hhkp\notag\\
     &\xrightarrow[\text{ Distribution}]{\text{In}} \ybf \herm \Ckpkp^{0.5}\Dbf^{-1} \Ckk\herm \Dbf^{-1} \Ckpkp^{0.5} \ybf , 
     \notag\\
      &\xrightarrow[]{\text{Lemma \ref{Lemma 2}}} \tr~\!  \Ckpkp \Dbf^{-1} \Ckk \Dbf^{-1}, 
\end{align}
where $\xbf$ and $\ybf$ are i.i.d. $\CN\left(\zero, \Ibf_{M\Nrf}\right)$. Substituting the \eqref{eq:fa4} and \eqref{eq:fa2} in \eqref{eq:fa3}, completes the proof.

\textbf{Second Approximation:}
Consider the SINR expression in \eqref{eq:upSINRexp}. 
\begin{align}
    \label{eq:hkohk}
     \Pup_k \hhk\herm \vect{\Omega}_k^{-1}\hhk &\xrightarrow[\text{ Distribution}]{\text{In}} \Pup_k \xbf \herm \Ckk^{0.5}\vect{\Omega}_k^{-1}\Ckk^{0.5} \xbf,
    \notag\\
     &\xrightarrow[]{\text{Lemma \ref{Lemma 2}}} \Pup_k \tr~\! \Ckk^{0.5}\vect{\Omega}_k^{-1} \Ckk^{0.5} = \tr~\! \Ckk\vect{\Omega}_k^{-1}\notag\\
     &\xrightarrow[]{\text{Lemma \ref{Lemma 3}}} \Pup_k \tr~\! \Ckk\vect{\Omega}^{-1}\notag\\
     &\xrightarrow[]{\text{Theorem \ref{Theorem 1}}} \Pup_k \frac{1}{M\Nrf} \tr~\! \Ckk\Tbf,
\end{align}
where $\xbf \sim \CN\left(\zero, \Ibf_{M\Nrf}\right)$, $\vect{\Omega}$, $\vect{\Omega}_k$, and $\Tbf$ are as in \eqref{eq:SINRvec}, \eqref{eq:upSINRexp} and \eqref{eq:Tval}, respectively. Note that to use Theorem \ref{Theorem 1} here, we again assume that the matrix $\Ckkp$ is close to zero.

\section{ Proof of Thm.~\ref{thm:pOptup}}
\label{app:pOptup}

For our proof we need the following definitions

\begin{Definition}[\textbf{Competitive Utility Functions} {\cite[Assumption~1]{hong2014unified}}]
Function ${\rm u}(\vect{p}): \mathbb{R}^K \mapsto \mathbb{R}^K$ is a competitive utility function if
\begin{enumerate}
    \item Positivity: For all $k$, ${\rm u}(\vect{p}) > \zero$ if $\vect{p}> \zero$, and ${\rm u}_k(\vect{p}) = 0$ if and only if $p_k = 0$ \footnote{All the inequalities are meant for elementwise}.
    \item Competitiveness: For all $k, \in[K]$, ${\rm u}_k$ is strictly increasing with respect to $p_k$ and is strictly decreasing with respect to $p_{k'},~ k'\neq k$ when $p_k > 0$.
    \item Directional Monotonicity: For $\lambda > 1$ and $\pbf >0 $, ${\rm u}(\lambda \pbf) > {\rm u}(\pbf)$.
\end{enumerate}
\end{Definition}

\begin{Definition}[\textbf{Monotonic Constraints} {\cite[Assumption~2]{hong2014unified}}]
Function ${\rm c}(\vect{p}): \mathbb{R}^K \mapsto \mathbb{R}^K$ is a monotonic constraints function if
\begin{enumerate}
    \item Strict Monotonicity: For all $k$, ${\rm c}_k(\pbf_1) > {\rm c}_k(\pbf_2)$ if $\pbf_1>\pbf_2$, and ${\rm u}_k(\pbf_1) \geq {\rm u}_k(\pbf_2)$ if $\pbf_1 \geq \pbf_2$.
    \item Feasibility: The set $\{\pbf > \zero : {\rm c}(\pbf) \leq  \bar{\cbf} \}$ is non-empty.
    \item Validity: For any $\pbf > \zero$, there exists $\lambda > 0$ such that
${\rm c}_k(\lambda \pbf) \geq \bar{c}_k$, for some $k$.
\end{enumerate}
\end{Definition}
It is easy to show that the users power constraints satisfy the definition of monotonic constraints. Here, we prove that the SINR expressions satisfy the competitive utility functions definition. 

Let us consider the SINR expression in \eqref{eq:upSINRexp}. 

\textbf{First condition:} Positivity is satisfied as $\vect{\Omega}_k$ is summation of positive definite matrices and so $\vect{\Omega}_k^{-1}$ is also positive definite. Therefore, $\SINRup_k >0$ unless $\Pup_k = 0$. 

\textbf{Second condition:} To prove this condition, we show that$\frac{\partial \SINRup_k}{\partial \Pup_k} >0$ and $\frac{\partial \SINRup_k}{\partial \Pup_{k'}} <0,~ k' \neq k$.  For $\frac{\partial \SINRup_k}{\partial \Pup_k} >0$,
\begin{align}
    &\frac{\partial \SINRup_k}{\partial \Pup_k} = \frac{\partial}{\partial \Pup_k} \hhk\herm \Pup_k\vect{\Omega}_k^{-1} \hhk \\
    &\stackrel{(a)}{=} \hhk\herm \frac{\partial \Gbf_k}{\partial \Pup_k} \hhk\\
    &\stackrel{(b)}{=} -\hhk\herm \Gbf_k \frac{\partial \Gbf_k^{-1}}{\partial \Pup_k} \Gbf_K \hhk,
\end{align}
where 
\begin{align}
    &\Gbf_k = \left(\!\sum_{k'\neq k}\!\! \frac{\Pup_{k'}}{\Pup_{k}}\!\! \left(\hhkp\hhkp\herm \!\! +\!\! \Rbf_{k'}\!\! -\!\! \Cbf_{{k'}}\!\!\right)\!\!+ \!\!\Rbf_{k}\!\!-\!\! \Cbf_{{k}}\!\!+\!\! \frac{1}{\Pup_k}\Cbf_{z}\!\!\right)^{-1}\!\!\!\!\!\!\!\!,\\
    &\frac{\partial \Gbf_k^{-1}}{\partial \Pup_k}\!\! =\! - \left(\!\sum_{k'\neq k} \!\!\frac{\Pup_{k'}}{{\Pup_{k}}^2} \!\!\left(\hhkp\hhkp\herm \!\! +\!\! \Rbf_{k'}\!\! - \!\!\Cbf_{{k'}}\!\! \right)\!\!+ \!\!\frac{1}{{\Pup_k}^2}\Cbf_z\!\!\right),
\end{align}
(a) follows from substituting the value of $\Dbf$ from \eqref{eq:Dval} in the expression of $\vect{\Omega}_k$ form \eqref{eq:upSINRexp} and (b) follows for the equality $\frac{\partial \Abf^{-1}}{\partial x} = -\Abf^{-1}\frac{\partial \Abf}{\partial x}\Abf^{-1}$. Note that based on the above expression, $\frac{\partial \Gbf_k^{-1}}{\partial \Pup_k}$ is negative definite and $\Gbf_k$ is positive definite. Define $\ybf =\Gbf_K \hhk $, then $\frac{\partial \SINRup_k}{\partial \Pup_k} = - \ybf\herm \frac{\partial \Gbf_k^{-1}}{\partial \Pup_k} \ybf  >0$. Next, we show $\frac{\partial \SINRup_k}{\partial \Pup_{k'}} <0,~ k' \neq k$. 
\begin{align}
    \frac{\partial \SINRup_k}{\partial \Pup_{k'}} &= \Pup_k\hhk \herm \frac{\partial \vect{\Omega}_k}{\partial \Pup_{k'}}\hhk\notag\\
    &\stackrel{(b)}{=} -\Pup_k \hhk \herm \vect{\Omega}_k\left( \hhkp\hhkp\herm +  \Rbf_{k} - \Cbf_{{k}} \right) \vect{\Omega}_k\hhk. 
\end{align}
Note that $\vect{\Omega}_k$ and $\left( \hhkp\hhkp\herm +  \Rbf_{k} - \Cbf_{{k}} \right)$ are positive definite. Define $\ybf =\vect{\Omega}_k\hhk$, then, $\frac{\partial \SINRup_k}{\partial \Pup_{k'}} = -\Pup_k \ybf \herm \left( \hhkp\hhkp\herm +  \Rbf_{k} - \Cbf_{{k}} \right) \ybf< 0$.

\textbf{Third condition:} To prove this condition, let us consider the difference
\begin{align}
    &\SINRup_k(\lambda\mathbf{\Pup}) - \SINRup_k(\mathbf{\Pup}) =\notag\\
    &\Pup_k\hhk\herm \left(\left(\frac{\vect{\Omega}_k(\lambda\mathbf{\Pup})}{\lambda}\right)^{-1}  - \vect{\Omega}_k(\mathbf{\Pup})^{-1}\right) \hhk
\end{align}
using Woodbury inversion lemma we have
\begin{align}
    &\left(\frac{\vect{\Omega}_k(\lambda\mathbf{\Pup})}{\lambda}\right)^{-1}  - \vect{\Omega}_k(\mathbf{\Pup})^{-1} =\notag\\
    &\frac{\vect{\Omega}_k(\lambda\mathbf{\Pup})}{\lambda} - \frac{\vect{\Omega}_k(\lambda\mathbf{\Pup})}{\lambda}\left(\frac{\vect{\Omega}_k(\lambda\mathbf{\Pup})}{\lambda} -  \vect{\Omega}_k(\mathbf{\Pup})\right)\frac{\vect{\Omega}_k(\lambda\mathbf{\Pup})}{\lambda}
    \notag\\
    &= \frac{\vect{\Omega}_k(\lambda\mathbf{\Pup})}{\lambda} + \frac{\vect{\Omega}_k(\lambda\mathbf{\Pup})}{\lambda}\left(  \vect{\Omega}_k(\mathbf{\Pup}) -
    \frac{\vect{\Omega}_k(\lambda\mathbf{\Pup})}{\lambda}
    \right)\frac{\vect{\Omega}_k(\lambda\mathbf{\Pup})}{\lambda}
\end{align}
It is easy to check that both $\frac{\vect{\Omega}_k(\lambda\mathbf{\Pup})}{\lambda}$ and $\vect{\Omega}_k(\mathbf{\Pup}) - \frac{\vect{\Omega}_k(\lambda\mathbf{\Pup})}{\lambda}$ are positive definite matrices. Let us denote them by matrices $\Abf$ and $\Bbf$ respectively. We have 
\begin{align}
    \SINRup_k(\lambda\mathbf{\Pup}) - \SINRup_k(\mathbf{\Pup}) =\notag \\
    \Pup_k \hhk \herm \Abf \hhk + \Pup_k \hhk \herm \Abf \Bbf \Abf\hhk. 
\end{align}
Similar steps as previous conditions, we can show that both terms are positive and
\begin{align}
     \SINRup_k(\lambda\mathbf{\Pup}) - \SINRup_k(\mathbf{\Pup}) > 0.
\end{align}
This concludes the proof for SINR expression in \eqref{eq:upSINRexp}.

To prove that \eqref{eq:ur2} is also a competitive utility function, a similar set of steps as of the proof for \eqref{eq:upSINRexp} can be used plus the well known results that for $\Abf$ positive definite $\tr~\! \Abf \Bbf >0$ when $\Bbf$ is positive definite and $\tr~\! \Abf \Bbf <0$ when $\Bbf$ is negative definite.

\section{ Proof of Thm.~\ref{thm:rmdl}}
\label{app:rmdl}
To prove this theorem, we use a similar procedure as of the one used for \cite[Theorem.~5]{6415388}.
Consider the SINR expression in \eqref{eq:SINRvec}. We perfom approximation of each term separately. First we approximate  $\Exp\{|\hat{\vbf}_{mk}^d|^2\}$ as. 

\begin{align}
\label{eq:vkd}
    &|\hat{\vbf}_{k}|^2 = \hhk\herm \left(\sum_{k'} \hhkp\hhkp\herm \right)^{-2}\hhk, \notag\\
     &\xrightarrow[]{\text{Lemma \ref{Lemma 1}}}
    \frac{\hhk\herm \left(\sum_{k'\neq k} \hhkp\hhkp\herm \right)^{-2}\hhk}{\left(1+\hhk
    \herm\left(\sum_{k'\neq k} \hhkp\hhkp\herm  \right)^{-1}\hhk\right)^2}
\end{align}
using Lemma \ref{Lemma 2}, Lemma \ref{Lemma 3}, and Theorem \ref{Theorem 1} for the nominator and Lemma \ref{Lemma 2}, Lemma \ref{Lemma 3}, and Theorem \ref{Theorem 2} for the denominator, we have 
\begin{align}
    &|\hat{\vbf}_{k}|^2 \approx \frac{ \frac{1}{M^2\Nrf^2}\tr~\! \Ckk\Sbf'}{\left(1 + \frac{1}{M\Nrf}\tr~\! \Ckk\Sbf\right)^2},
\end{align}
where $\Sbf$ and $\Sbf'$ are as in \eqref{eq:sbf} and  \eqref{eq:sbfp}, respectively. Note that, the matrices $\Sbf'$ and $\Ckk$ are block diagonal with $M$ block matrices of size $N_{RF}\times N_{RF}$ each only depending to the large scale statistics corresponding to one of the APs. Moreover, we have $|\vbf_{k}|^2 = \sum_{m\in[M]}|\vbf_{mk}^d|^2$. Rewriting the left and right hand side of \eqref{eq:vkd}, we have
\begin{align}
\label{eq:vkd2}
    \sum_{m\in[M]}|\hat{\vbf}_{mk}^d|^2 \approx \sum_{m\in[M]} \frac{ \frac{1}{M^2\Nrf^2}  \tr~\! \Ckk[m]\Sbf'[m]}{\left(1 + \frac{1}{M\Nrf}\tr~\! \Ckk\Sbf\right)^2},
\end{align}
where $\Ckk[m]$ and $\Sbf'[m]$ are the $m^{\rm th}$ block matrices of size $N_{RF}\times N_{RF}$ located on the diagonal of $\Ckk$ and $\Tbf'$, respectively. Based on the one-to-one correspondence observed between the left and right hand side of \eqref{eq:vkd2}, we approximate the $|\hat{\vbf}_{mk}^d|^2$ as
\begin{align}
    |\hat{\vbf}_{mk}^d|^2\approx \frac{ \frac{1}{M^2\Nrf^2}  \tr~\! \Ckk[m]\Sbf'[m]}{\left(1 + \frac{1}{M\Nrf}\tr~\! \Ckk\Sbf\right)^2},
\end{align}
Let us define $\nu_{k} = \max_m \sqrt{|\hat{\vbf}_{mk}^d|^2}$. Next, we approximate $ | \Exp \{\hbf_k\herm \vbf_k^d\} |^2 $.
\begin{align}
\label{eq:n1}
    \vbf_k^{d\herms}\hbf_k &= \frac{1}{\nu_k}\hhk \herm \left(\sum_{k'}\hhkp\hhkp\herm  + \rho\Ibf \right)^{-1} \hdek\notag\\ &\xrightarrow[]{\text{Lemma \ref{Lemma 1}}} \frac{1}{\nu_k} \frac{\hhk \herm \vect{\Omega}_{\rm RZF,k}^{-1}\hdek}{1+  \hhk\herm \vect{\Omega}_{\rm RZF,k}^{-1} \hhk},
\end{align}
where $\vect{\Omega}_{\rm RZF,k}$ is as in \eqref{eq:dcomb} where from the summation we remove the element $k' = k$. Using Lemma~\ref{Lemma 2}, Lemma~\ref{Lemma 3}, and Thm.~\ref{Theorem 1}, we have
\begin{align}
\label{eq:n2}
    \hhk\herm \vect{\Omega}_{\rm RZF,k}^{-1} \hhk \approx \frac{1}{M\Nrf} \tr~\! \Ckk\Sbf,
\end{align}
Substituting \eqref{eq:n2} in \eqref{eq:n1}, we get 
\begin{align}
\label{eq:ntot}
        |  \Exp \{\hk \herm \vbf_k^d \} |^2  \approx \frac{1}{\nu_k^2} \frac{(\frac{1}{M\Nrf}\tr~\! \Ckk \Sbf )^2}{\left(1+  \frac{1}{M\Nrf}\tr~\! \Ckk\Sbf\right)^2}  
\end{align}
Next, we approximate the term $\sum_{k' = 1}^K{\mathbb{E}}\left\{| \hbf_k\herm  \vbf_{k'}^d |^2\right\}$.
\begin{align}
\label{eq:d1}
    \left|\vbf_{k'}^{d\herms}\hbf_k\right|^2&\xrightarrow[]{\text{Lemma \ref{Lemma 1}}}\frac{1}{\nu_{k'}^2} \left|\frac{\hhkp \herm\vect{\Omega}_{\rm RZF,k'}^{-1}\hbf_k}{1+  \hhkp\herm \vect{\Omega}_{\rm RZF,k'}^{-1} \hhkp}\right|^2.
\end{align}
We already have an approximation for  $ \hhk\herm \vect{\Omega}_{\rm RZF,k'}^{-1} \hhk$ from \eqref{eq:n2}. As for the nominator, we are in fact interested in its absolute value to power two.
\begin{align}
    \label{eq:d2}
    &\left|\hhkp \herm \vect{\Omega}_{\rm RZF,k'}^{-1}\hbf_k\right|^2 = \left|\hhkp \herm \vect{\Omega}_{\rm RZF,k'}^{-1}\hbf_k \hbf_k\herm\vect{\Omega}_{\rm RZF,k'}^{-1} \hhkp \right|\notag\\
    &\xrightarrow[]{\text{Lemma \ref{Lemma 2}}} \frac{1}{M \Nrf}\left|\hbf_k\herm\vect{\Omega}_{\rm RZF,k'}^{-1}\Ckpkp \vect{\Omega}_{\rm RZF,k'}^{-1}\hbf_k \right|.
\end{align}
Substituting
\begin{align}
    \vect{\Omega}_{\rm RZF,k'}^{-1} = \vect{\Omega}_{\rm RZF,k,k'}^{-1} - \frac{\vect{\Omega}_{\rm RZF,k,k'}^{-1}\hhk \hhk\herm \vect{\Omega}_{\rm RZF,k,k'}^{-1}}{1+ \hhk \vect{\Omega}_{\rm RZF,k,k'}^{-1}\hhk},
\end{align}
where $\vect{\Omega}_{\rm RZF,k,k'} = \sum_{k'' \neq k,k'} \hhkpp\hhkpp\herm  + \rho\Ibf$, we have 
\begin{align}
    \begin{aligned}
\hk\herm &\vect{\Omega}_{\rm RZF,k'}^{-1} \Ckpkp \vect{\Omega}_{\rm RZF,k'}^{-1}\hk =  \hk\herm \vect{\Omega}_{\rm RZF,k,k'}^{-1} \Ckpkp \vect{\Omega}_{\rm RZF,k,k'}^{-1}\hk \notag\\
&+ \frac{\left|\hk\herm \vect{\Omega}_{\rm RZF,k,k'}^{-1} \hhk\right|^2\hhk \herm \vect{\Omega}_{\rm RZF,k,k'}^{-1}\Ckpkp \vect{\Omega}_{\rm RZF,k,k'}^{-1} \hhk }{(1+ \hhk \vect{\Omega}_{\rm RZF,k,k'}^{-1}\hhk)^2}\\
&- 2{\rm Re}\left\{\frac{\hhk \herm  \vect{\Omega}_{\rm RZF,k,k'}^{-1} \hk \hk \herm \vect{\Omega}_{\rm RZF,k,k'}^{-1}\Ckpkp \vect{\Omega}_{\rm RZF,k,k'}^{-1} \hhk}{1+ \hhk \vect{\Omega}_{\rm RZF,k,k'}^{-1}\hhk}\right\}.
    \end{aligned}
\end{align}
Applying Lemma~\ref{Lemma 2}, Lemma~\ref{Lemma 3}, and Thm.~\ref{Theorem 2}, we have
\begin{subequations}
\label{eq:d3}
\begin{align}
    &\hk\herm\vect{\Omega}_{\rm RZF,k,k'}^{-1} \Ckpkp \vect{\Omega}_{\rm RZF,k,k'}^{-1}\hk \approx \frac{1}{M^2\Nrf^2}\tr~\! \Rbf_k \Sbf'_{k'}\\
    &\hhk \herm \vect{\Omega}_{\rm RZF,k,k'}^{-1}\Ckpkp\vect{\Omega}_{\rm RZF,k,k'}^{-1} \hhk \approx \frac{1}{M^2\Nrf^2} \tr~\! \Ckk \Sbf'_{k'}\\
    &\hk \herm \vect{\Omega}_{\rm RZF,k,k'}^{-1}\Ckpkp \vect{\Omega}_{\rm RZF,k,k'}^{-1} \hhk \approx \frac{1}{M^2\Nrf^2} \tr~\! \Ckk \Sbf'_{k'}
\end{align}
\end{subequations}
where $ \mathbf{S}'_{k'} $ is as in \eqref{eq:sbfp1}.
Similarly, using Lemma~\ref{Lemma 2}, Lemma~\ref{Lemma 3}, and Thm.~\ref{Theorem 1}, we have
\begin{subequations}
\label{eq:d4}
\begin{align}
    \hhk \herm  \vect{\Omega}_{\rm RZF,k,k'}^{-1} \hk \approx \frac{1}{M\Nrf}\tr~\! \Ckk\Sbf\\
    \hhk \herm  \vect{\Omega}_{\rm RZF,k,k'}^{-1} \hhk \approx \frac{1}{M\Nrf}\tr~\! \Ckk\Sbf.
\end{align}
\end{subequations}
Substituting \eqref{eq:d4} and \eqref{eq:d3} in \eqref{eq:d2} and the result in \eqref{eq:d1}, we have
\begin{align}
    \label{eq:dtot}
    \begin{aligned}
    &{\mathbb{E}}\left\{| \hk \herm \vbf^d_{k'} |^2\right\} \approx \\
    &\frac{1}{\nu_{k'}^2(M\Nrf + \tr~\! \Ckpkp \Sbf)^2}\times\Big[ \tr~\!\Rbf_k\Sbf'_{k'} \\
    &- 2{\rm Re}\left\{\frac{\tr~\!\Ckk \Sbf\, \tr~\! \Ckk\Sbf'_{k'}}{M\Nrf + \tr~\! \Ckk \Sbf}\right\}\\
    &+(M\Nrf + \tr~\! \Ckk \Sbf)^2(\tr~\! \Ckk \Sbf )^2 \,\tr~\! \Ckk \Sbf'_{k'}\Big ]
    \end{aligned}
\end{align}
Also, we approximate $ {\mathbb{V}}\left\{ \hk\herm \vbf^d_k \right\} \approx 0$. Therefore, using \eqref{eq:ntot} and \eqref{eq:dtot} in \eqref{eq:dlSINRexp}, completes our derivation.
\end{document}